\documentclass[onecolumn,superscriptaddress,aps]{revtex4-2}
\expandafter\let\csname equation*\endcsname\relax
\expandafter\let\csname endequation*\endcsname\relax
\usepackage{amsmath}
\usepackage{graphicx}
\usepackage{amsthm}
\usepackage{amssymb}
\usepackage{latexsym}
\usepackage{array}
\usepackage{hyperref}
\usepackage{amsfonts}
\usepackage{dsfont}
\usepackage{algpseudocode}
\usepackage{mathrsfs}
\usepackage{verbatim}
\usepackage{bbold}
\usepackage{lipsum}
\usepackage[normalem]{ulem}

\usepackage{upgreek}

\usepackage{makecell}
\usepackage{adjustbox,lipsum}

\usepackage{algorithm2e}
\usepackage{algpseudocode}
\usepackage{placeins}

\usepackage{color}

\usepackage{bm}
\usepackage{times}

\newcommand{\bra}[1]{\left<#1\right|}
\newcommand{\ket}[1]{\left|#1\right>}
\newcommand{\abs}[1]{\left|#1\right|}

\newcommand{\braket}[2]{\left<{#1}|{#2}\right>}
\newcommand{\ketbra}[2]{\ket{#1}\!\!\bra{#2}}

\newtheorem{proposition}{Proposition}

\begin{document}

\title{First-Quantized Relativistic Quantum Simulation with Periodic and Dirichlet Boundary Conditions}

\author{Jeongho Bang}
\address{Institute for Convergence Research and Education in Advanced Technology, Yonsei University, Seoul 03722, Republic of Korea}
\address{Department of Quantum Information, Yonsei University, Incheon 21983, Republic of Korea}

\author{Timothy P. Spiller}
\address{York Centre for Quantum Technologies, School of Physics, Engineering and Technology, University of York, York, YO10 5DD, UK}

\author{Kyunghyun Baek}
\address{Institute for Convergence Research and Education in Advanced Technology, Yonsei University, Seoul 03722, Republic of Korea}
\address{Department of Quantum Information, Yonsei University, Incheon 21983, Republic of Korea}

\author{Kyoungho Cho}
\address{Institute for Convergence Research and Education in Advanced Technology, Yonsei University, Seoul 03722, Republic of Korea}
\address{Department of Statistics and Data Science, Yonsei University, Seoul 03722, Republic of Korea}

\author{Jaewoo Joo}
\address{School of Computing, Mathematics and Physics, University of Portsmouth, Portsmouth PO1 3QL, UK}

\date{\today}

\begin{abstract}
In this work, we present a methodology for first-quantized relativistic quantum simulation on one-dimensional finite domains under the two boundary conditions most commonly used in lattice models: periodic boundary conditions (PBC) and Dirichlet boundary conditions (DBC). Starting from the positive-energy relativistic kinetic operator, we construct weakly relativistic lattice Hamiltonians whose leading correction requires the boundary-consistent discretized momentum moments \(\langle \hat{P}^{2}\rangle\) and \(\langle \hat{P}^{4}\rangle\). These moments are reconstructed in the PBC Hamiltonian from moments of a unitary cyclic translation while the DBC Hamiltonian uses the open-chain finite-difference. In a qubit-register implementation, it can be evaluated as the corresponding cyclic translation estimator plus boundary-local terms that remove the unphysical wrap-around link. The resulting energy-estimation workflow uses translation measurements for the kinetic terms,
a small number of endpoints and near-endpoints overlap probabilities for DBC, and position-basis sampling for diagonal potentials. 
We validate the framework of the relativistic quantum simulation in various benchmark potentials such as no potential and a cosine potential for PBC as well as an infinite square well and a harmonic potential for DBC, with finite-shot sampling tests. These benchmarks show good agreement between the estimator reconstruction and direct matrix evaluation while separating the finite-grid discretization, weak-relativistic truncation, and measurement errors. 
\end{abstract}

\maketitle


\section{Introduction}\label{sec:intro}

Quantum simulation is a promising approach to continuum quantum dynamics, where a useful simulation must represent not only a Hilbert space but also the Hamiltonian that defines the physical problem~\cite{Feynman1982,Lloyd1996,Georgescu2014,Altman2021,Daley2022,Fauseweh2024}. In this setting, relativistic quantum simulation is a natural target because the kinetic energy is no longer simply quadratic in momentum. Even for a single positive-energy particle, the relativistic kinetic operator is a nonlinear function of the momentum, and therefore a lattice simulator must specify how the corresponding momentum moments are represented on a finite domain~\cite{Bauer2023,DiMeglio2024,Sakurai2020,Gerritsma2010,Li2023}.

The relativistic kinetic term considered here is the positive-energy single-particle operator
\begin{eqnarray}
\hat{T}_{\rm rel} &=& mc^2 \left( \sqrt{\hat{\mathds{1}} + \frac{\hat{p}^{2}}{m^2 c^2}} - \hat{\mathds{1}} \right) = - m c^2 \sum_{k=1}^{\infty} \chi_k \, \hat{p}^{2k},
\label{eq:intro_positive_energy_operator} \end{eqnarray}
for 
$\chi_k = \frac{(-1)^{k} (2k-3)!!}{2^k k! \, \left( m\, c \right)^{2k}}$. In the weakly relativistic regime, this has the expansion
\begin{eqnarray}
\hat{T}_{\rm rel} = \frac{\hat{p}^{2}}{2m} - \frac{\hat{p}^{4}}{8 m^3 c^2} + O\left( \frac{\hat{p}^{6}}{m^5 c^4} \right),
\label{eq:intro_relativistic_expansion}
\end{eqnarray}
for continuum relativistic momentum operator $\hat{p}$.
Thus, a leading-order relativistic quantum simulation must estimate not only the first kinetic energy term with $\langle\hat{p}^{2}\rangle$, but also the second one $\langle\hat{p}^{4}\rangle$~\cite{NoteLatticeMoments}. Therefore, the main question addressed in this work is how to define and estimate those moments consistently when the finite-domain relativistic Hamiltonian is subject to either periodic boundary conditions (PBC) or Dirichlet boundary conditions (DBC)~\cite{Sakurai2020,Childs2022RealSpace,Costa2019Wave,Childs2021PDE}.

The boundary conditions are part of the definition of the finite-domain Hamiltonian, rather than a detail to be added after discretization. Under PBC, the spatial domain is closed into a circle, so the finite-difference link connecting the last grid point back to the first is physical. Under DBC, the wavefunction vanishes at the endpoints of an interval, and the physical finite-difference operator is an open-chain operator with no wrap-around link. 
A relativistic quantum simulator for PBC and DBC therefore corresponds to different lattice momentum operators and, consequently, to different moments. Thus, here we treat PBC and DBC as two boundary-condition-resolved relativistic quantum-simulation problems.

First-quantized grid encodings provide the common qubit-register language. For an $L$-qubit position register, a 1D wavefunction is represented as $\ket{\psi} = \sum_{j=0}^{2^L - 1} c_j\ket{j}$, where the computational basis state $\ket{j}$ labels a grid point. In this representation, diagonal potentials are evaluated from the position-basis probabilities, while the derivative operators are expressed through finite-difference shift operators on the register~\cite{Costa2019Wave,Childs2021PDE,Kassal2008,Babbush2018LinearT,Babbush2019,Su2021,Babbush2023,Berry2024,Georges2025,Childs2022RealSpace,Lubasch2020,JJ2021,JooSpiller2023}. The qubit-encoding issue is that the natural shift operation on an $L$-qubit register is cyclic. This is exactly the right unitary operation for PBC, but for DBC the same cyclic shift contains an unphysical wrap-around coupling. In the methodology developed below, this implementation issue is handled after the target PBC or DBC relativistic Hamiltonian has been fixed: the cyclic translation moments are used to evaluate the bulk finite-difference contribution, while DBC is enforced by boundary-local correction terms which can remove the wrap-around link.

With the above-described perspectives, we develop the first-quantized relativistic quantum-simulation framework under PBC and DBC. We first define the weakly relativistic lattice Hamiltonian for each boundary condition, with the fourth-order kinetic moment understood as the square of the same boundary-consistent momentum-squared operator. We then derive measurement estimators for the energy terms. The PBC case requires translation moments up to second order, and the DBC one uses the same translation moments together with a finite set of endpoint probabilities and endpoint or near-endpoint coherences. The diagonal potentials are estimated by position sampling and classical post-processing. In this way, the qubit-register treatment of boundary conditions supports our relativistic quantum simulation under PBC and DBC.

Our framework for relativistic quantum simulation is validated analytically and numerically. A periodic free particle tests the PBC lattice dispersion and its convergence to the continuum relativistic energy. A Dirichlet infinite square well tests the open-chain DBC spectrum and the boundary-local reconstruction of the kinetic moments that enter the first and second leading relativistic kinetic-energy terms. The smooth-potential ground states verify the full energy-estimation workflow with diagonal potentials, and finite-shot sampling tests confirm the expected inverse-square-root statistical scaling with the number of measurement shots. These benchmarks keep distinct the finite-grid discretization error, the weak-relativistic truncation error, and the statistical error of the estimator.

The contributions of this work are threefold. First, we formulate first-quantized weakly relativistic lattice Hamiltonians for finite-domain quantum simulation under PBC and DBC. Second, we derive observable estimators that reconstruct the second-order and fourth-order lattice momentum moments required by the weak-relativistic expansion from translation measurements and, for DBC, boundary-local overlap measurements. Third, we validate the construction with analytically solvable PBC and DBC systems, smooth-potential benchmarks, and finite-shot simulations. Together, these results establish a compact methodology for relativistic quantum simulation with PBC and DBC.

\section{Relativistic Lattice Hamiltonians under Periodic and Dirichlet Boundary Conditions}\label{sec:lattice_hamiltonian}

We now define the lattice Hamiltonians that serve as the targets for relativistic simulations with PBC and DBC. The same first-quantized position register is used in both cases, but the finite-difference momentum operator is chosen to match the physical boundary condition. We denote the continuum momentum operator by $\hat{p}$ and the corresponding lattice momentum operator by $\hat{P}_{\tau}$ with $\tau \in \{{\rm P}, {\rm D}\}$ specifying PBC or DBC. 

\subsection{First-quantized grids and weakly relativistic kinetic terms}\label{subsec:grid_rel_expansion}

We consider a single particle in one spatial dimension on a finite interval of physical length $R$. An $L$-qubit position register represents $N=2^L$ grid amplitudes, with a general lattice state written as
\begin{eqnarray}
\ket{\psi} = \sum_{j=0}^{N-1} c_j \ket{j},
\label{eq:lattice_state}
\end{eqnarray}
where $\sum_{j=0}^{N-1} \abs{c_j}^2 = 1$. Here, the computational basis state $\ket{j}$ labels a spatial grid point. The map from $j$ to the physical coordinate is part of the finite-domain simulation specification and differs for PBC and DBC.

For PBC, the grid points are taken as, for $j=0,\ldots,N-1$,
\begin{eqnarray}
x_j^{({\rm P})} = j\Delta_{\rm P},
\quad
\Delta_{\rm P} = \frac{R}{N}.
\label{eq:pbc_grid}
\end{eqnarray}
The last point is connected back to the first point, so that the finite difference is naturally represented by a cyclic translation on the register.

For DBC, we instead use $N$ interior grid points, for $j=0,\ldots,N-1$, 
\begin{eqnarray}
x_j^{({\rm D})} = (j+1)\Delta_{\rm D},
\quad
\Delta_{\rm D} = \frac{R}{N+1},
\label{eq:dbc_grid}
\end{eqnarray}
together with the virtual boundary values, 
\begin{eqnarray}
\psi(0)=0,
\quad
\psi(R)=0 .
\label{eq:dbc_virtual_boundaries}
\end{eqnarray}
This convention implements DBC as an open-chain finite-difference operator. The endpoints of the computational register, $\ket{0}$ and $\ket{N-1}$, therefore represent the first and last interior grid points, not the physical boundary points themselves. The corresponding physical-grid and dimensional conventions are summarized in \ref{app:grid_conventions}.

In the weakly relativistic regime, where the relevant states have momentum support small compared with $mc$, we use the expansion in Eq.~(\ref{eq:intro_relativistic_expansion}). 
On the lattice, the corresponding boundary-dependent kinetic estimator through the leading relativistic correction is given by  
\begin{eqnarray}
\hat{T}_{\tau} = \frac{\hat{P}^{2}_{\tau}}{2m} - \frac{\hat{P}^{4}_{\tau}}{8 m^3 c^2},
\quad
\tau \in \{{\rm P}, {\rm D}\},
\label{eq:T_rel_lattice_tau}
\end{eqnarray}
where $\tau={\rm P}$ and $\tau={\rm D}$ denote PBC and DBC, respectively. Here and below,
\begin{eqnarray}
\hat{P}^{4}_{\tau} \equiv \left(\hat{P}^{2}_{\tau}\right)^2.
\label{eq:P4_definition}
\end{eqnarray}
Thus, the relativistic correction is built from the same boundary-consistent lattice momentum-squared operator used for the first kinetic term.

For a scalar potential $V(x)$, the lattice potential is diagonal in the
position basis,
\begin{eqnarray}
\hat{V}_{\tau} = \sum_{j=0}^{N-1} V\bigl(x_j^{(\tau)}\bigr) \ketbra{j}{j}.
\label{eq:lattice_potential_sec2}
\end{eqnarray}
The corresponding total lattice Hamiltonian used for energy estimation is therefore
\begin{eqnarray}
\hat{H}^{tot}_{\tau} &=& \hat{T}_{\tau} + \hat{V}_{\tau}.
\label{eq:lattice_hamiltonian_tau}
\end{eqnarray}


\subsection{Relativistic kinetic lattice Hamiltonian in PBC} \label{subsec:pbc_lattice_hamiltonian}

For PBC, we introduce the unitary cyclic translation operator $\hat{A}$ known as a quantum adder/subtractor \cite{Vedral1996},
\begin{eqnarray}
\hat{A}\ket{j} &=& \ket{j+1 \ {\rm mod} \ N}, \nonumber \\
\hat{A}^{\dagger}\ket{j} &=& \ket{j-1 \ {\rm mod} \ N}.
\label{eq:cyclic_translation}
\end{eqnarray}
It satisfies
\begin{eqnarray}
\hat{A}^{\dagger}\hat{A} = \hat{A}\hat{A}^{\dagger} = \hat{\mathds{1}},
\quad
\hat{A}^{N} = \hat{\mathds{1}}.
\label{eq:cyclic_translation_unitarity}
\end{eqnarray}
The standard second-order finite-difference representation of the momentum-squared operator is given by a lattice Laplacian operator then
\begin{eqnarray}
\hat{P}^{2}_{\rm P} = - \frac{\hbar^2}{\Delta_{\rm P}^{2}} \left( \hat{A} + \hat{A}^{\dagger} - 2\hat{\mathds{1}} \right).
\label{eq:P2_pbc_operator}
\end{eqnarray}
For a lattice state $\ket{\psi}$, this gives
\begin{eqnarray}
\langle \hat{P}^{2}_{\rm P} \rangle = \frac{2 \hbar^2}{\Delta_{\rm P}^{2}} \left( 1- \mathrm{Re} \langle\hat{A}\rangle \right),
\label{eq:P2_pbc_expectation}
\end{eqnarray}
where $\mathrm{Re} \langle\hat{A}\rangle = \left( \langle \hat{A} \rangle + \langle \hat{A}^{\dagger} \rangle\right)/2$.
The fourth-order momentum moment is obtained by squaring Eq.~(\ref{eq:P2_pbc_operator}). Since $\hat{A}$ is unitary,
\begin{eqnarray}
\hat{P}^{4}_{\rm P} = \frac{\hbar^4}{\Delta_{\rm P}^{4}} \left( \hat{A}^{2} + \hat{A}^{\dagger 2} - 4\hat{A} - 4\hat{A}^{\dagger} + 6\hat{\mathds{1}} \right),
\label{eq:P4_pbc_operator}
\end{eqnarray}
and, therefore,
\begin{eqnarray}
\langle \hat{P}^{4}_{\rm P} \rangle = \frac{2 \hbar^4}{\Delta_{\rm P}^{4}} \left( \mathrm{Re} \langle\hat{A}^2\rangle - 4 \mathrm{Re} \langle\hat{A}\rangle + 3 \right).
\label{eq:P4_pbc_expectation}
\end{eqnarray}
Eqs~(\ref{eq:P2_pbc_expectation}) and~(\ref{eq:P4_pbc_expectation}) show that the leading relativistic correction under PBC requires only translation moments up to second order. In particular, $\langle\hat{P}^{2}_{\rm P}\rangle$ depends on $\langle\hat{A}\rangle$, while $\langle\hat{P}^{4}_{\rm P}\rangle$ depends on both $\langle\hat{A}\rangle$ and $\langle\hat{A}^{2}\rangle$.

The PBC kinetic energy through order $\hat{P}^{4}$ is thus
\begin{eqnarray}
\hspace{-1.5cm} \langle \hat{T}_{\rm P} \rangle = \frac{1}{2m} \langle \hat{P}^{2}_{\rm P} \rangle - \frac{1}{8 m^3 c^2} \langle \hat{P}^{4}_{\rm P} \rangle.
\label{eq:T_pbc_expectation}
\end{eqnarray}
This expression is the PBC kinetic component of the relativistic simulation framework and will be used in the next Sec.~\ref{sec:quantum_estimators} to define the translation-moment measurement protocol.

\subsection{Relativistic kinetic lattice Hamiltonian in DBC} \label{subsec:dbc_lattice_hamiltonian}

For DBC, the target finite-difference Hamiltonian is the open-chain Hamiltonian associated with the interior grid in Eq.~(\ref{eq:dbc_grid}). In a qubit-register implementation, it is convenient to express this open-chain operator using the same unitary cyclic translation $\hat{A}$, supplemented by boundary-local terms that remove the wrap-around link. On the $N$-point interior grid, we define
\begin{eqnarray}
\hat{E}_0 = \ketbra{N-1}{0} + \ketbra{0}{N-1}.
\label{eq:E0_definition}
\end{eqnarray}
This operator $\hat{E}_0$ identifies the coupling between the first and last interior grid points that would be present in a cyclic register translation but is absent from the DBC Hamiltonian.

The DBC momentum-squared operator is therefore
\begin{eqnarray}
\hat{P}^{2}_{\rm D} = -\frac{\hbar^2}{\Delta_{\rm D}^{2}} \left( \hat{A} + \hat{A}^{\dagger} - 2\hat{\mathds{1}} - \hat{E}_0 \right).
\label{eq:P2_dbc_operator}
\end{eqnarray}
Equivalently, if
\begin{eqnarray}
\hat{P}^{2}_{\rm cyc}(\Delta_{\rm D}) = -\frac{\hbar^2}{\Delta_{\rm D}^{2}} \left( \hat{A} + \hat{A}^{\dagger} - 2\hat{\mathds{1}} \right)
\label{eq:P2_cyc_dgrid}
\end{eqnarray}
denotes the cyclic finite-difference operator evaluated with the DBC grid
spacing, then
\begin{eqnarray}
\hat{P}^{2}_{\rm D} = \hat{P}^{2}_{\rm cyc}(\Delta_{\rm D}) + \frac{\hbar^2}{\Delta_{\rm D}^{2}} \hat{E}_0.
\label{eq:P2_dbc_correction_operator}
\end{eqnarray}
Hence,
\begin{eqnarray}
\langle \hat{P}^{2}_{\rm D} \rangle = \langle \hat{P}^{2}_{\rm cyc}(\Delta_{\rm D}) \rangle + \frac{\hbar^2}{\Delta_{\rm D}^{2}} \langle \hat{E}_0 \rangle,
\label{eq:P2_dbc_expectation}
\end{eqnarray}
where
\begin{eqnarray}
\langle \hat{E}_0 \rangle = c_{N-1}^{\ast} c_0 + c_0^{\ast} c_{N-1}.
\label{eq:E0_expectation}
\end{eqnarray}

The fourth-order DBC moment is defined as
\begin{eqnarray}
\hat{P}^{4}_{\rm D} = \left( \hat{P}^{2}_{\rm D} \right)^2.
\label{eq:P4_dbc_definition}
\end{eqnarray}
Using Eq.~(\ref{eq:P2_dbc_operator}), one obtains
\begin{eqnarray}
\langle \hat{P}^{4}_{\rm D} \rangle &=& \langle \hat{P}^{4}_{\rm cyc}(\Delta_{\rm D}) \rangle  + \frac{\hbar^4}{\Delta_{\rm D}^{4}} \left( 4\langle \hat{E}_0\rangle - \langle \hat{E}_1\rangle - \langle \hat{E}_2\rangle + \langle \hat{E}_0^2\rangle \right),
\label{eq:P4_dbc_expectation}
\end{eqnarray}
where
\begin{eqnarray}
\hat{P}^{4}_{\rm cyc}(\Delta_{\rm D}) = \frac{\hbar^4}{\Delta_{\rm D}^{4}} \left( \hat{A}^{2} + \hat{A}^{\dagger 2} - 4\hat{A} - 4\hat{A}^{\dagger} + 6\hat{\mathds{1}} \right),
\label{eq:P4_cyc_dgrid}
\end{eqnarray}
and the additional boundary operators are
\begin{eqnarray}
\hat{E}_1 = \hat{A}\hat{E}_0 + \hat{E}_0\hat{A}^{\dagger},
\quad
\hat{E}_2 = \hat{E}_0\hat{A} + \hat{A}^{\dagger}\hat{E}_0.
\label{eq:E1_E2_definitions}
\end{eqnarray}
The derivation of Eq.~(\ref{eq:P4_dbc_expectation}) follows by expanding $(\hat{A}+\hat{A}^{\dagger}-2\hat{\mathds{1}}-\hat{E}_0)^2$ and collecting the cyclic part and the boundary-local part. The full algebra
is given in \ref{app:boundary_moments}.

For later use, we record the explicit endpoint form of the boundary operators. For $N>2$,
\begin{eqnarray}
\hat{E}_0^2 = \ketbra{0}{0} + \ketbra{N-1}{N-1},
\label{eq:E0_squared_explicit}
\end{eqnarray}
and the compact definitions in Eq.~(\ref{eq:E1_E2_definitions}) give
\begin{eqnarray}
\hat{E}_1 &=& 2\ketbra{0}{0} + \ketbra{1}{N-1} + \ketbra{N-1}{1}, \nonumber \\
\hat{E}_2 &=& 2\ketbra{N-1}{N-1} + \ketbra{0}{N-2} + \ketbra{N-2}{0}.~~~
\label{eq:E12_explicit}
\end{eqnarray}
Accordingly,
\begin{eqnarray}
\langle \hat{E}_0^2\rangle &=& \abs{c_0}^2 + \abs{c_{N-1}}^2, \nonumber \\
\langle \hat{E}_1\rangle &=& 2\abs{c_0}^2 + c_1^{\ast} c_{N-1} + c_{N-1}^{\ast} c_1, \nonumber \\
\langle \hat{E}_2\rangle &=& 2\abs{c_{N-1}}^2 + c_0^{\ast} c_{N-2} + c_{N-2}^{\ast} c_0.
\label{eq:E012_expectation}
\end{eqnarray}
Thus, the DBC part of the implementation is boundary-local: it involves only endpoint probabilities and a small number of endpoint or near-endpoint coherences.

The DBC kinetic energy through order $\hat{P}^{4}$ is
\begin{eqnarray}
\langle \hat{T}_{\rm D} \rangle = \frac{1}{2m} \langle \hat{P}^{2}_{\rm D} \rangle - \frac{1}{8m^3c^2} \langle \hat{P}^{4}_{\rm D} \rangle.
\label{eq:T_dbc_expectation}
\end{eqnarray}
Combining Eqs.~(\ref{eq:P2_dbc_expectation}) and (\ref{eq:P4_dbc_expectation}), this can be written as a cyclic translation-moment contribution plus a boundary correction:
\begin{eqnarray}
\hspace{-2cm} \langle \hat{T}_{\rm D} \rangle &=&  \langle \hat{T}_{\rm cyc}(\Delta_{\rm D}) \rangle + \frac{\hbar^2}{2m\Delta_{\rm D}^{2}} \langle \hat{E}_0\rangle  - \frac{\hbar^4}{8m^3c^2\Delta_{\rm D}^{4}} \left( 4\langle \hat{E}_0\rangle - \langle \hat{E}_1\rangle - \langle \hat{E}_2\rangle + \langle \hat{E}_0^2\rangle \right).~~~~~~~
\label{eq:T_dbc_correction_form}
\end{eqnarray}
This form completes the DBC Hamiltonian construction. The bulk contribution is evaluated by the same translation moments used for the PBC simulator, while the difference between DBC and a cyclic register translation is captured by a finite set of boundary-local expectation values.

\section{Quantum Estimators for PBC and DBC Relativistic Simulation}\label{sec:quantum_estimators}

We now convert the PBC and DBC relativistic lattice Hamiltonians of Sec.~\ref{sec:lattice_hamiltonian} into observable estimators. The aim is to specify the energy-evaluation layer of the quantum simulation, independent of any particular state-preparation ansatz. Once a first-quantized lattice state $\ket{\psi}$ as in Eq.~(\ref{eq:lattice_state}) has been prepared, the weakly relativistic energy through $\hat{P}^{4}$ is reconstructed from a small set of expectation values. PBC uses translation moments of the cyclic shift operator. DBC uses the same translation moments and adds the boundary-local overlap probabilities required by the open-chain Hamiltonian.

Throughout this section, an estimator obtained from finitely many shots will be denoted by a wide-tilde, for example $\widetilde{m}_l$. The statistical analysis of the finite-shot fluctuations is given in \ref{app:shot_noise}.

\subsection{PBC translation-moment estimator}\label{subsec:translation_moment_estimator}

For PBC, the kinetic moments in Eqs.~(\ref{eq:P2_pbc_expectation}) and~(\ref{eq:P4_pbc_expectation}) depend on the real parts of $\langle\hat{A}\rangle$ and $\langle\hat{A}^2\rangle$. Here, we define
\begin{eqnarray}
m_l = \mathrm{Re}\bra{\psi}\hat{A}^{l}\ket{\psi},
\quad
l=1,2.
\label{eq:ml_definition}
\end{eqnarray}
Since $\langle \hat{A}^{l} + \hat{A}^{\dagger l} \rangle = 2\mathrm{Re} \langle\hat{A}^{l}\rangle = 2 m_l$, the PBC momentum moments can be written as
\begin{eqnarray}
\langle \hat{P}_{\rm P}^{2}\rangle &=& \frac{2\hbar^2}{\Delta_{\rm P}^{2}} \left( 1 - m_1 \right),
\label{eq:P2_pbc_m1}
\\
\langle \hat{P}_{\rm P}^{4}\rangle &=& \frac{\hbar^4}{\Delta_{\rm P}^{4}} \left( 2m_2 - 8m_1 + 6 \right).
\label{eq:P4_pbc_m1_m2}
\end{eqnarray}
Thus, the leading relativistic kinetic energy under PBC in Eq.~(\ref{eq:T_pbc_expectation}) is reconstructed as
\begin{eqnarray}
\hspace{-2cm} \langle \hat{T}_{\rm P} \rangle &=& \alpha \left( 1 - m_1 \right) - \frac{\beta}{4} \left( m_2 - 4m_1 + 3 \right)
= \left( \alpha - \frac{3\beta}{4} \right) + 
\left( \beta - \alpha \right) m_1
- \frac{ \beta}{4} m_2 .
\label{eq:T_pbc_moments}
\end{eqnarray}
where $\alpha = \hbar^2/(m \Delta^2_{\rm P})$ and $\beta = \hbar^4/(m^3 c^2 \Delta^4_{\rm P})$.

The moments $m_l$ are obtained from the standard controlled-unitary interferometric primitive~\cite{Ekert2002,Alves2003}. We state the identity explicitly because it is the basic observable estimator used throughout this work.

\begin{proposition}[Translation-moment Hadamard estimator]
\label{prop:hadamard_translation}
Let $\hat{U}$ be a unitary operator acting on the system register and let $\ket{\psi}$ be the system state. Initialise an ancilla qubit in $\ket{0}$, apply a Hadamard gate $H$, apply controlled-$\hat{U}$, apply a second Hadamard gate to the ancilla, and measure the ancilla in the $Z$ basis. Then,
\begin{eqnarray}
\langle \hat{Z}_C \rangle = \mathrm{Re}\bra{\psi}\hat{U}\ket{\psi}.
\label{eq:hadamard_estimator_identity}
\end{eqnarray}
In particular, choosing $\hat{U}=\hat{A}^{l}$ gives $\langle\hat{Z}_C\rangle=m_l$.
\end{proposition}

\begin{proof}---After the first Hadamard gate and the controlled-$\hat{U}$ operation, the joint state is
\begin{eqnarray}
\ket{\Phi} = \frac{1}{\sqrt{2}} \left( \ket{0}_C\ket{\psi} + \ket{1}_C\hat{U}\ket{\psi} \right),
\label{eq:hadamard_state_after_controlled_U}
\end{eqnarray}
and 
\begin{eqnarray}
\langle\hat{Z}_C\rangle &=& \frac{1}{2} \left( \bra{\psi}\hat{U}\ket{\psi} + \bra{\psi}\hat{U}^{\dagger}\ket{\psi} \right) = \mathrm{Re}\bra{\psi}\hat{U}\ket{\psi}.
\label{eq:hadamard_proof}
\end{eqnarray}
This proves the claim.
\end{proof}

\begin{figure}[t]
\centering
\includegraphics[width=0.7\linewidth]{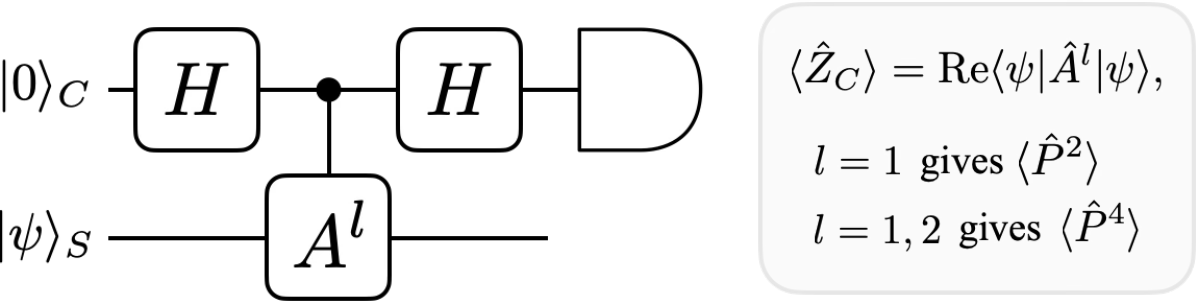}
\caption{Quantum circuit of the translation-moment estimator. A Hadamard-test-type circuit estimates the real part of a translation moment, $m_l=\mathrm{Re}\langle\hat{A}^{l}\rangle$. The moment $l=1$ is sufficient for the first kinetic term, while the second leading relativistic correction requires $l=1,2$.}
\label{fig:TM-estimator}
\end{figure}

The quantum circuit of the translation-moment estimator is drawn in Fig.~\ref{fig:TM-estimator}. Here, if the imaginary part of a translation moment is required, one may equivalently measure the ancilla in the $Y$ basis, or insert the appropriate phase rotation before the final Hadamard measurement. The present kinetic estimators require only the real parts because $\hat{P}^{2}$ and $\hat{P}^{4}$ are Hermitian and depend on $\hat{A}^{l}+\hat{A}^{\dagger l}$.

With $M_l$ shots allocated to the $l$th translation moment, the empirical estimator is
\begin{eqnarray}
\widetilde{m}_l = \frac{1}{M_l} \sum_{r=1}^{M_l} z_r^{(l)},
\quad
z_r^{(l)} \in \{-1,+1\},
\label{eq:empirical_translation_moment}
\end{eqnarray}
and Eqs.~(\ref{eq:P2_pbc_m1})--(\ref{eq:T_pbc_moments}) are evaluated by replacing $m_l$ with $\widetilde{m}_l$.

\subsection{DBC boundary-overlap estimator}\label{subsec:boundary_overlap_estimator}

The DBC correction terms in Eqs.~(\ref{eq:P2_dbc_expectation}) and~(\ref{eq:P4_dbc_expectation}) are boundary-local. They contain endpoint probabilities and coherences between a small number of endpoint or near-endpoint basis states. These coherences have the generic form
$\left\langle \bigl( \ketbra{f}{g} + \ketbra{g}{f} \bigr) \right\rangle,
\label{eq:generic_boundary_coherence}
$ with computational basis states $\ket{f}$ and $\ket{g}$. The following identity gives a direct probability estimator for such terms.

\begin{proposition}[Boundary-overlap identity]
\label{prop:boundary_overlap_identity}
Let $f \neq g$ and define
\begin{eqnarray}
P_f &=& \abs{\braket{f}{\psi}}^2, \nonumber \\
P_g &=& \abs{\braket{g}{\psi}}^2, \nonumber \\
P_{fg}^{+} &=& \abs{\langle s_{fg}^{+}|\psi\rangle}^2, 
\label{eq:P_f_g_fg}
\end{eqnarray}
where $| s_{fg}^{+} \rangle = (\ket{f}+\ket{g})/\sqrt{2}$. Then,
\begin{eqnarray}
\left\langle \bigl( \ketbra{f}{g} + \ketbra{g}{f} \bigr) \right\rangle = 2P_{fg}^{+} - P_f - P_g.
\label{eq:boundary_overlap_identity}
\end{eqnarray}
\end{proposition}

\begin{proof}---Writing $c_f=\braket{f}{\psi}$ and $c_g=\braket{g}{\psi}$, the normalisation of $|s_{fg}^{+}\rangle$ gives
\begin{eqnarray}
P_{fg}^{+} = \abs{\frac{c_f+c_g}{\sqrt{2}}}^2 = \frac{\abs{c_f}^2 + \abs{c_g}^2 + c_f^{\ast} c_g + c_g^{\ast} c_f}{2}.
\label{eq:Pfg_plus_expansion}
\end{eqnarray}
Therefore,
\begin{eqnarray}
2P_{fg}^{+} - P_f - P_g &=& c_f^{\ast} c_g + c_g^{\ast} c_f = \left\langle \bigl( \ketbra{f}{g} + \ketbra{g}{f} \bigr) \right\rangle.
\label{eq:Pfg_plus_expansion_final}
\end{eqnarray}
The proof is completed.
\end{proof}

\begin{figure}[t]
\centering
\includegraphics[width=0.45\linewidth]{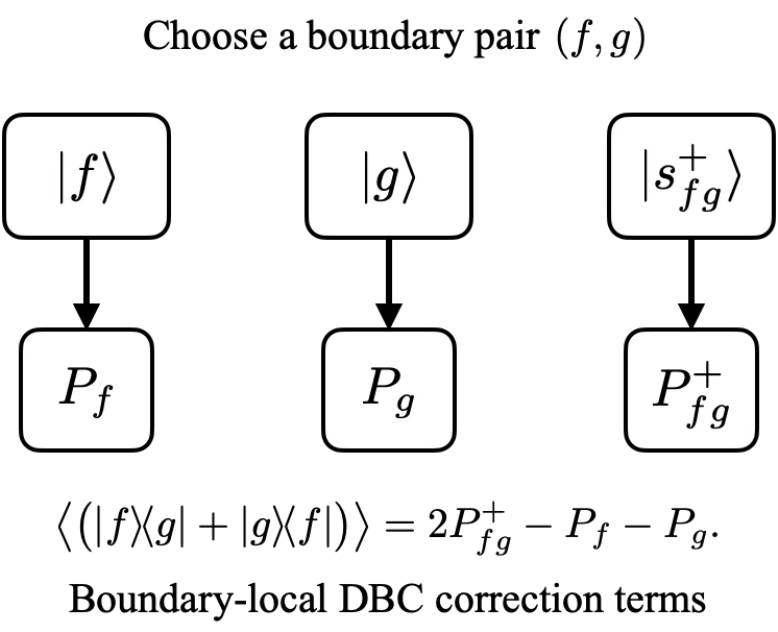}
\caption{Schematic for the boundary-overlap estimator. DBC introduces boundary-local coherences. For a boundary pair $(f, g)$, the probabilities associated with $\ket{f}$, $\ket{g}$, and the normalised superposition $|s_{fg}^{+}\rangle=(\ket{f}+\ket{g})/\sqrt{2}$ determine $\left\langle \bigl( \ketbra{f}{g} + \ketbra{g}{f} \bigr) \right\rangle = 2P_{fg}^{+}-P_f-P_g$. These terms reconstruct the DBC boundary corrections $\hat{E}_0,\hat{E}_1,\hat{E}_2$, and $\hat{E}_0^2$.}
\label{fig:BO-estimator}
\end{figure}

The factor of two in Eq.~(\ref{eq:boundary_overlap_identity}) is essential: it follows from the normalisation of $|s_{fg}^{+}\rangle=(\ket{f}+\ket{g})/\sqrt{2}$. The identity holds for general
complex amplitudes and returns the real coherence $2\mathrm{Re}(c_f^{\ast}\, c_g)$, which is precisely the quantity needed for the Hermitian boundary operators. Schematic of the boundary-overlap estimator is given in Fig.~\ref{fig:BO-estimator}. The corresponding phase-dependent identities are collected in \ref{app:boundary_overlap}.

For compactness, we define the pair-coherence estimator
\begin{eqnarray}
B_{f, g} \equiv 2P_{fg}^{+} - P_f - P_g.
\label{eq:Bfg_definition}
\end{eqnarray}
Then, using the explicit DBC boundary operators in Eqs.~(\ref{eq:E0_squared_explicit})--(\ref{eq:E12_explicit}), their expectation values are
\begin{eqnarray}
b_0 &\equiv& \langle\hat{E}_0\rangle = B_{0,N-1}, \nonumber \\
b_{00} &\equiv& \langle\hat{E}_0^2\rangle = P_0 + P_{N-1}, \nonumber \\
b_1 &\equiv& \langle\hat{E}_1\rangle = 2P_0 + B_{1,N-1}, \nonumber \\
b_2 &\equiv& \langle \hat{E}_2\rangle = 2P_{N-1} + B_{0,N-2}.
\label{eq:b_estimators}
\end{eqnarray}
Here, $P_j=\abs{\braket{j}{\psi}}^2$. We assume $N>2$ so that the endpoint and near-endpoint labels appearing in Eq.~(\ref{eq:b_estimators}) are distinct. The small grids $N=1,2$ are degenerate as finite-difference representations of a Dirichlet interval and are not used in the numerical benchmarks.

The DBC momentum moments are therefore reconstructed from the same translation moments $m_1, m_2$ and the four boundary quantities $b_0, b_1, b_2, b_{00}$:
\begin{eqnarray}
\langle \hat{P}_{\rm D}^{2}\rangle &=& \frac{\hbar^2}{\Delta_{\rm D}^{2}} \left( 2 - 2 m_1 + b_0 \right),
\label{eq:P2_dbc_estimator_moments}
\\
\langle \hat{P}_{\rm D}^{4}\rangle &=& \frac{\hbar^4}{\Delta_{\rm D}^{4}} \left( 2m_2 - 8m_1 + 6 + 4b_0 - b_1 - b_2 + b_{00} \right). \label{eq:P4_dbc_estimator_moments} 
\end{eqnarray}
The corresponding relativistic kinetic-energy estimator is
\begin{eqnarray}
\langle\hat{T}_{\rm D} \rangle = \frac{1}{2m} \langle\hat{P}_{\rm D}^{2}\rangle - \frac{1}{8 m^3 c^2} \langle\hat{P}_{\rm D}^{4}\rangle .
\label{eq:T_dbc_estimator_moments}
\end{eqnarray}

For reference, Table~\ref{tab:boundary_estimators} lists the boundary terms needed through the leading relativistic correction. The table is written in terms of measurable probabilities rather than amplitudes, and therefore directly specifies the estimator inputs.

\begin{table}[ht]
\centering
\begin{tabular}{c c c c}
Boundary quantity & Pair state(s) & Probability combination & Used-in \\
\hline
$b_0=\langle \hat{E}_0\rangle$		& $(0,N-1)$	& $B_{0,N-1}$ 			& $\hat{P}_{\rm D}^{2},\hat{P}_{\rm D}^{4}$ \\
$b_{00}=\langle \hat{E}_0^2\rangle$ & $(0,N-1)$		& $P_0+P_{N-1}$ 		& $\hat{P}_{\rm D}^{4}$ \\
$b_1=\langle \hat{E}_1\rangle$		& $(1,N-1)$	& $2P_0+B_{1,N-1}$	 	& $\hat{P}_{\rm D}^{4}$ \\
$b_2=\langle \hat{E}_2\rangle$		& $(0,N-2)$	& $2P_{N-1}+B_{0,N-2}$	& $\hat{P}_{\rm D}^{4}$
\end{tabular}
\caption{Boundary-local probability estimators required for DBC through $\hat{P}^{4}$. For a pair $(f,g)$, $|s_{fg}^{+}\rangle=(\ket{f}+\ket{g})/\sqrt{2}$ and $B_{f,g}=2P_{fg}^{+}-P_f-P_g$. The diagonal endpoint probabilities are obtained by computational-basis measurements.}
\label{tab:boundary_estimators}
\end{table}

\subsection{Potential-energy estimator and total-energy workflow}\label{subsec:potential_and_workflow}

The potential term is simpler than the kinetic term because the lattice potential is diagonal in the computational basis. To compute the potential energy $\langle \hat{V}_{\tau} \rangle$, we may use a quantum circuit introduced in \cite{Lubasch2020,JJ2021}. This scheme requires the initial preparation of a control qubit, wavefunction qubits and a diagonal density matrix $\hat{\rho}_V= \sum_{j = 0}^{2^L-1}   {\cal V}_j  \ket{j}\bra{j}$, because $\hat{V} = {\cal S}\, \hat{\rho}_V  =  {\cal S}\,\sum_{j=0}^{2^L-1} {\cal V}_j \ket{j}\bra{j}$ for the scale factor of the potential ${\cal S}$. Then, one performs a controlled block-SWAP gate between the wavefunction qubits and the density matrix. 

However, we here examine a direct and practical method of position-sampling estimation for the potential energy.
For the boundary conditions $\tau \in \{{\rm P},{\rm D}\}$, Eq.~(\ref{eq:lattice_potential_sec2}) gives
\begin{eqnarray}
\langle\hat{V}_{\tau}\rangle = \sum_{j=0}^{N-1} V\bigl( x_j^{(\tau)} \bigr) \abs{c_j}^2 .
\label{eq:potential_position_average}
\end{eqnarray}
Thus, $\langle\hat{V}_{\tau}\rangle$ is obtained by measuring the position register in the computational basis and applying a classical post-processing function $j \mapsto V\bigl(x_j^{(\tau)}\bigr)$.
\begin{proposition}[Position-sampling estimator for the potential]
\label{prop:potential_sampling}
Let $j_1,\ldots,j_M$ be independent computational-basis samples from $\ket{\psi}$, so that $\Pr[j_r=j] = \abs{c_j}^2$. Define
\begin{eqnarray}
\widetilde{V}_{\tau} = \frac{1}{M}\sum_{r=1}^{M} V\bigl( x_{j_r}^{(\tau)} \bigr).
\label{eq:potential_sampling_estimator}
\end{eqnarray}
Then, $\widetilde{V}_{\tau}$ is an unbiased estimator of $\langle\hat{V}_{\tau}\rangle$: i.e., 
\begin{eqnarray}
\mathbb{E}\bigl[ \widetilde{V}_{\tau} \bigr] = \langle\hat{V}_{\tau}\rangle .
\label{eq:potential_unbiased}
\end{eqnarray}
If the sampled potential values lie in $[V_{\min}, V_{\max}]$, then
\begin{eqnarray}
\mathrm{Var}\bigl[ \widetilde{V}_{\tau} \bigr] \le \frac{(V_{\max}-V_{\min})^2}{4M}.
\label{eq:potential_variance_bound}
\end{eqnarray}
\end{proposition}

\begin{proof}---By taking the expectation value of Eq.~(\ref{eq:potential_sampling_estimator}) over the measurement outcomes, we attain
\begin{eqnarray}
\mathbb{E}\bigl[ \widetilde{V}_{\tau} \bigr] = \sum_{j=0}^{N-1} V\bigl(x_j^{(\tau)}\bigr) \abs{c_j}^2 = \langle \hat{V}_{\tau}\rangle .
\label{eq:potential_unbiased_proof}
\end{eqnarray}
The variance bound follows from the fact that the variance of any random variable supported on an interval of width $V_{\max} - V_{\min}$ is at most $(V_{\max} - V_{\min})^2/4$, together with the $1/M$ variance reduction for the sample mean.
\end{proof}

The total energy through the leading relativistic correction is then
\begin{eqnarray}
\langle\hat{H}_{\tau}^{tot}\rangle = \langle \hat{T}_{\tau} \rangle + \langle\hat{V}_{\tau}\rangle,
\quad
\tau \in \{{\rm P}, {\rm D}\}.
\label{eq:total_energy_estimator_sec3}
\end{eqnarray}
The workflow is summarised in Algorithm~\ref{alg:energy_estimator}. The same procedure can be used either as a post-processing estimator for a fixed state or as the energy-evaluation subroutine inside a variational optimisation loop~\cite{Peruzzo2014,Cerezo2021,Bharti2022,Tilly2022}. In the latter case, the variational parameters determine the prepared state $\ket{\psi(\bm{\theta})}$, but the estimator identities themselves are unchanged.

\begin{algorithm}[ht]
\caption{BC-resolved estimator for $\langle\hat{H}_{\tau}^{tot}\rangle$.}
\label{alg:energy_estimator}
\begin{algorithmic}[1]
\Require Boundary condition (BC) $\tau \in \{{\rm P}, {\rm D}\}$, Grid size $N=2^L$, State-preparation circuit for $\ket{\psi}$, Potential values $V(x_j^{(\tau)})$.
\State Prepare the $L$-qubit first-quantised lattice state $\ket{\psi}=\sum_j c_j\ket{j}$.
\State Estimate $m_1=\mathrm{Re}\langle\hat{A}\rangle$ using the translation-moment circuit.
\State Estimate $m_2=\mathrm{Re}\langle\hat{A}^2\rangle$ using the translation-moment circuit.
\If{$\tau = \mathrm{P}$}
    \State Reconstruct $\langle\hat{P}_{\rm P}^{2}\rangle$ and $\langle \hat{P}_{\rm P}^{4}\rangle$ from Eqs.~(\ref{eq:P2_pbc_m1}) and (\ref{eq:P4_pbc_m1_m2}).
\ElseIf{$\tau = \mathrm{D}$}
    \State Estimate the boundary quantities $b_0,b_1,b_2,b_{00}$ using Table~\ref{tab:boundary_estimators}.
    \State Reconstruct $\langle\hat{P}_{\rm D}^{2}\rangle$ and $\langle \hat{P}_{\rm D}^{4}\rangle$ from Eqs.~(\ref{eq:P2_dbc_estimator_moments}) and (\ref{eq:P4_dbc_estimator_moments}).
\State Compute $\langle\hat{T}_{\tau} \rangle =\langle\hat{P}_{\tau}^{2}\rangle/(2m) -  \langle\hat{P}_{\tau}^{4}\rangle/(8 m^3 c^2)$.
\State Estimate $\langle\hat{V}_{\tau}\rangle = \sum_j V\bigl(x_j^{(\tau)}\bigr)\abs{c_j}^2$ by computational-basis sampling and classical post-processing.
\State Return
$\langle\hat{H}_{\tau}^{tot}\rangle = \langle\hat{T}_{\tau} \rangle + \langle\hat{V}_{\tau}\rangle$.
\end{algorithmic}
\end{algorithm}

The important structural point is that PBC and DBC fit into the same relativistic energy-estimation workflow. The translation moments encode the cyclic finite-difference contribution to the kinetic energy, while DBC adds only the boundary-local information needed to realize the open-chain Hamiltonian. Consequently, the leading relativistic correction does not require a qualitatively new measurement primitive beyond those already needed for first-quantized kinetic-energy estimation; it extends the measured observable set from $\langle\hat{A}\rangle$ to $\langle\hat{A}\rangle, \langle\hat{A}^{2}\rangle$ and, for DBC, a finite set of boundary-overlap terms.

\section{Validation of the PBC/DBC Relativistic Simulation Framework}\label{sec:validation}

We now validate our PBC and DBC relativistic simulation framework developed in Secs.~\ref{sec:lattice_hamiltonian} and~\ref{sec:quantum_estimators}. The validation has two roles. First, analytically solvable systems test whether the PBC and DBC lattice Hamiltonians reproduce the expected spectra and continuum trends. Second, the estimator reconstruction is compared with direct matrix evaluation on the same lattice to verify the measurement formulas. This separation is important: the finite-difference lattice introduces a discretization error, the weakly relativistic expansion introduces a truncation error, and finite measurements introduce a statistical error. The matrix construction, estimator reconstruction from state vectors, benchmark parameters, and error metrics used for these validations are collected in \ref{app:numerical_methods}.

Throughout the numerical examples, we use natural units, i.e., $\hbar=m=1$, unless stated otherwise. The grid spacing is chosen according to the boundary condition, namely $\Delta_{\rm P}=R/N$ for PBC and $\Delta_{\rm D}=R/(N+1)$ for DBC. The relativistic energy reported below is
the perturbative estimator
\begin{eqnarray}
\langle\hat{T}_{\tau}\rangle = \frac{1}{2m}\langle\hat{P}^{2}_{\tau}\rangle - \frac{1}{8 m^3 c^2}\langle\hat{P}^{4}_{\tau}\rangle,
\quad
\tau \in \{{\rm P}, {\rm D}\}.
\label{eq:validation_T4}
\end{eqnarray}

\subsection{PBC benchmark: free-particle dispersion} \label{subsec:pbc_free_particle_validation}

For PBC, the natural analytic benchmark is the free particle. Define the lattice Fourier state: for $n=0,1,\ldots,N-1$,
\begin{eqnarray}
\ket{q_n} = \frac{1}{\sqrt{N}} \sum_{j=0}^{N-1} \exp\left(\frac{2\pi i n j}{N}\right) \ket{j}.
\label{eq:pbc_fourier_mode}
\end{eqnarray}
This state diagonalises the cyclic translation operator and therefore also the PBC finite-difference momentum operator.

\begin{proposition}[PBC lattice dispersion]
\label{prop:pbc_lattice_dispersion}
For the PBC momentum operator in Eq.~(\ref{eq:P2_pbc_operator}), the Fourier state $\ket{q_n}$ satisfies
\begin{eqnarray}
\hat{P}^{2}_{\rm P}\ket{q_n} &=& p^{2}_{\rm lat}(n)\ket{q_n}, \nonumber \\[2pt]
p^{2}_{\rm lat}(n) &=& \frac{4\hbar^2}{\Delta_{\rm P}^2} \sin^2\!\left(\frac{\pi n}{N}\right).
\label{eq:pbc_p2_lattice_dispersion}
\end{eqnarray}
In the continuum limit at fixed $R$ and fixed mode number $n$,
\begin{eqnarray}
p^{2}_{\rm lat}(n) ~\longrightarrow~ p^{2}_{\rm cont}(n) = \left(\frac{2\pi n\hbar}{R}\right)^2.
\label{eq:pbc_p2_continuum_dispersion}
\end{eqnarray}
\end{proposition}

\begin{proof}---Using $\hat{A}\ket{j}=\ket{j+1 \ {\rm mod} \ N}$, one obtains
\begin{eqnarray}
\hat{A}\ket{q_n} = \exp\left(-\frac{2\pi i n}{N}\right)\ket{q_n}.
\label{eq:pbc_A_eigenvalue_proof}
\end{eqnarray}
By substitution into $\hat{P}^{2}_{\rm P}=-(\hbar^2/\Delta_{\rm P}^{2})(\hat{A}+\hat{A}^{\dagger}-2\hat{\mathds{1}})$,  we have
\begin{eqnarray}
p^{2}_{\rm lat}(n) &=& -\frac{\hbar^2}{\Delta_{\rm P}^{2}} \left[ 2\cos\left(\frac{2\pi n}{N}\right)-2 \right] \nonumber \\
	&=& \frac{4\hbar^2}{\Delta_{\rm P}^{2}} \sin^2\left(\frac{\pi n}{N}\right).
\label{eq:pbc_dispersion_proof}
\end{eqnarray}
Since $\Delta_{\rm P}=R/N$ and $\sin(\pi n/N)=\pi n/N + O(N^{-3})$ for fixed $n$, the continuum limit follows.
\end{proof}

\begin{figure}[t]
\centering
\includegraphics[width=0.65\linewidth]{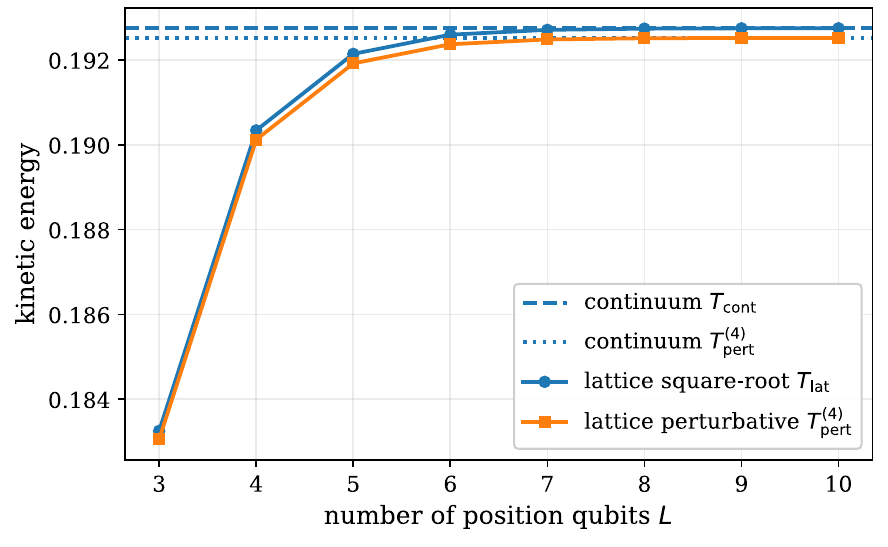}
\caption{PBC free-particle benchmark. The continuum relativistic energy $T_{\rm cont}$ is compared with the exact square-root lattice energy $T_{\rm lat}$ and the perturbative lattice energy $T_{\rm pert}^{(4)}$ for the Fourier mode $n=1$. Here, $R=10$ and $c=2$. The convergence of $T_{\rm lat}$ isolates the finite-grid discretization error, while the residual separation between $T_{\rm pert}^{(4)}$ and $T_{\rm cont}$ reflects the relativistic truncation error.}
\label{fig:pbc_free_particle_dispersion}
\end{figure}

For this benchmark, we compare three kinetic energies,
\begin{eqnarray}
T_{\rm cont} &=& mc^2\left[ \sqrt{1+\frac{p^{2}_{\rm cont}}{m^2c^2}} - 1 \right], \nonumber \\
T_{\rm lat} &=& mc^2\left[ \sqrt{1+\frac{p^{2}_{\rm lat}}{m^2c^2}} - 1 \right], \nonumber \\
T^{(4)}_{\rm pert}  &=& \frac{p^{2}_{\rm lat}}{2m} - \frac{p^{4}_{\rm lat}}{8m^3c^2}.
\label{eq:T_validation}
\end{eqnarray}
Fig.~\ref{fig:pbc_free_particle_dispersion} shows the result for $R=10$, $c=2$, and $n=1$, for which $p_{\rm cont}^2/(m^2c^2) \simeq 9.87 \times 10^{-2}$. The square-root lattice energy $T_{\rm lat}$ converges to the continuum value as $L$ increases. The perturbative value $T_{\rm pert}^{(4)}$ follows the same discretization trend but remains separated from $T_{\rm cont}$ by the controlled truncation error of the weakly relativistic expansion.

\subsection{DBC benchmark: infinite square well}\label{subsec:dbc_square_well_validation}

For DBC, the corresponding analytic benchmark is the infinite square well on the interval $[0,R]$. With the interior-grid convention of Eq.~(\ref{eq:dbc_grid}), the normalised discrete sine modes are, for $s=1,\ldots,N$, 
\begin{eqnarray}
\psi_s(j) = \sqrt{\frac{2}{N+1}} \sin\left[ \frac{\pi s(j+1)}{N+1} \right].
\label{eq:dbc_sine_modes}
\end{eqnarray}
They diagonalise the open-chain finite-difference Laplacian.

\begin{proposition}[DBC open-chain spectrum]
\label{prop:dbc_open_chain_spectrum}
For the DBC momentum operator in Eq.~(\ref{eq:P2_dbc_operator}), the sine mode in Eq.~(\ref{eq:dbc_sine_modes}) satisfies
\begin{eqnarray}
\hat{P}^{2}_{\rm D}\ket{\psi_s} &=& p^{2}_{{\rm D},\rm lat}(s)\ket{\psi_s}, \nonumber \\[2pt]
p^{2}_{{\rm D}, \rm lat}(s) &=& \frac{4\hbar^2}{\Delta_{\rm D}^{2}} \sin^2\left[\frac{\pi s}{2(N+1)}\right].
\label{eq:dbc_p2_lattice_spectrum}
\end{eqnarray}
Consequently,
\begin{eqnarray}
\bra{\psi_s}\hat{P}^{4}_{\rm D}\ket{\psi_s} = \bigl[ p^{2}_{{\rm D}, \rm lat}(s) \bigr]^2.
\label{eq:dbc_p4_lattice_spectrum}
\end{eqnarray}
\end{proposition}
\begin{figure*}[t]
\centering
\includegraphics[width=0.75\linewidth]{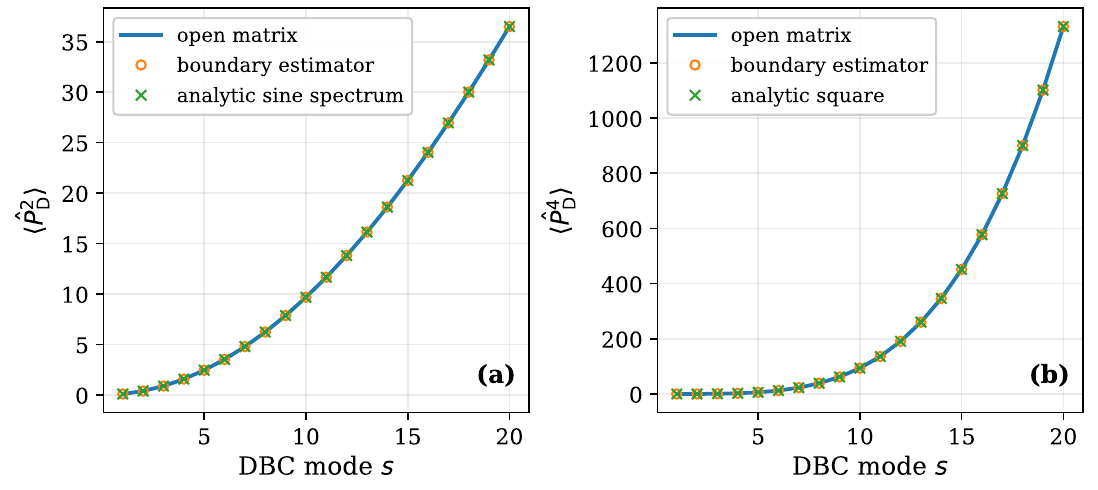}
\caption{DBC infinite-square-well benchmark. Here, we choose $N=64$ interior grid points and $R=10$. (a) We compare $\langle\hat{P}^{2}_{\rm D}\rangle$ obtained from the open-chain matrix evaluation, the boundary-corrected cyclic estimator, and the analytic sine spectrum. (b) We make the same comparison for $\langle\hat{P}^{4}_{\rm D}\rangle= \langle (\hat{P}^{2}_{\rm D})^2 \rangle$. The overlapping curves verify that the boundary-correction terms reconstruct the DBC finite-difference operator at the estimator level. }
\label{fig:dbc_square_well_validation}
\end{figure*}

\begin{proof}---The DBC finite-difference operator acts on the interior amplitudes as
\begin{eqnarray}
\bigl(\hat{P}^{2}_{\rm D}\psi\bigr)_j = \frac{\hbar^2}{\Delta_{\rm D}^{2}} \left(2\psi_j-\psi_{j-1}-\psi_{j+1}\right),
\label{eq:dbc_open_chain_action}
\end{eqnarray}
with the virtual values $\psi_{-1}=\psi_{N}=0$. By substituting the sine form in Eq.~(\ref{eq:dbc_sine_modes}) and using the elementary identity $2\sin\theta - \sin(\theta-\alpha) - \sin(\theta+\alpha) = 2(1-\cos\alpha)\sin\theta$, we attain the eigenvalue 
\begin{eqnarray}
\frac{2\hbar^2}{\Delta_{\rm D}^{2}}\left[ 1 - \cos\left(\frac{\pi s}{N+1}\right)\right],
\end{eqnarray}
which is equivalent to
Eq.~(\ref{eq:dbc_p2_lattice_spectrum}). Since
$\hat{P}^{4}_{\rm D}=(\hat{P}^{2}_{\rm D})^2$, Eq.~(\ref{eq:dbc_p4_lattice_spectrum}) follows.
\end{proof}

The estimator-level validation is stronger than the spectral check alone. For each sine mode, we compute $\langle\hat{P}^{2}_{\rm D}\rangle$ and $\langle\hat{P}^{4}_{\rm D}\rangle$ in three independent ways: direct open-chain matrix evaluation, the boundary-corrected cyclic estimator in Eqs.~(\ref{eq:P2_dbc_estimator_moments}) and (\ref{eq:P4_dbc_estimator_moments}), and the analytic eigenvalues above. 

Fig.~\ref{fig:dbc_square_well_validation} shows the agreement for $N=64$. The maximum relative discrepancy between the boundary-corrected estimator and the open-chain matrix evaluation is below $3.1 \times 10^{-14}$ for $\langle\hat{P}^{2}_{\rm D}\rangle$ and below $7.4 \times 10^{-11}$ for $\langle\hat{P}^{4}_{\rm D}\rangle$ over the displayed modes. This confirms that the DBC correction is not an approximation to the open-chain operator; it is an estimator decomposition of the same lattice operator.

\subsection{Smooth-potential benchmark}\label{subsec:smooth_potential_validation}

The previous two tests validate the kinetic operators in exactly solvable settings. We next include smooth potentials to test the full energy workflow on nontrivial lattice ground states. For PBC, we use the periodic potential
\begin{eqnarray}
V_{\rm P}(x) = V_0\left[ 1 - \cos\left(\frac{2\pi x}{R}\right) \right],
\label{eq:pbc_smooth_potential_validation}
\end{eqnarray}
while for DBC, we use a smooth confining potential in the box,
\begin{eqnarray}
V_{\rm D}(x) = \frac{1}{2}m\omega^2\left(x-\frac{R}{2}\right)^2.
\label{eq:dbc_smooth_potential_validation}
\end{eqnarray}
For each boundary condition and each grid size, we diagonalise the first-order total lattice Hamiltonian
\begin{eqnarray}
\hat{H}_{{\rm 1st},\tau} = \frac{\hat{P}^{2}_{\tau}}{2m} + \hat{V}_{\tau}
\label{eq:Hnr_validation}
\end{eqnarray}
and denote its ground state by $\ket{\psi_{0,\tau}}$. On this fixed state, we then evaluate
\begin{eqnarray}
E_{{\rm 1st},\tau} &=& \bra{\psi_{0,\tau}} \hat{H}_{{\rm 1st},\tau} \ket{\psi_{0,\tau}}, \nonumber \\
\Delta E_{{\rm rel},\tau} &=& - \frac{1}{8m^3c^2} \bra{\psi_{0,\tau}} \hat{P}^{4}_{\tau} \ket{\psi_{0,\tau}}, \nonumber \\
E_{{\rm rel},\tau}^{(4)} &=& E_{{\rm 1st},\tau} + \Delta E_{{\rm rel},\tau}.
\label{eq:E_validations}
\end{eqnarray}
Each quantity is evaluated both by direct matrix multiplication and by the estimator reconstruction of Sec.~\ref{sec:quantum_estimators}.

\begin{figure*}[t]
\centering
\includegraphics[width=0.7\linewidth]{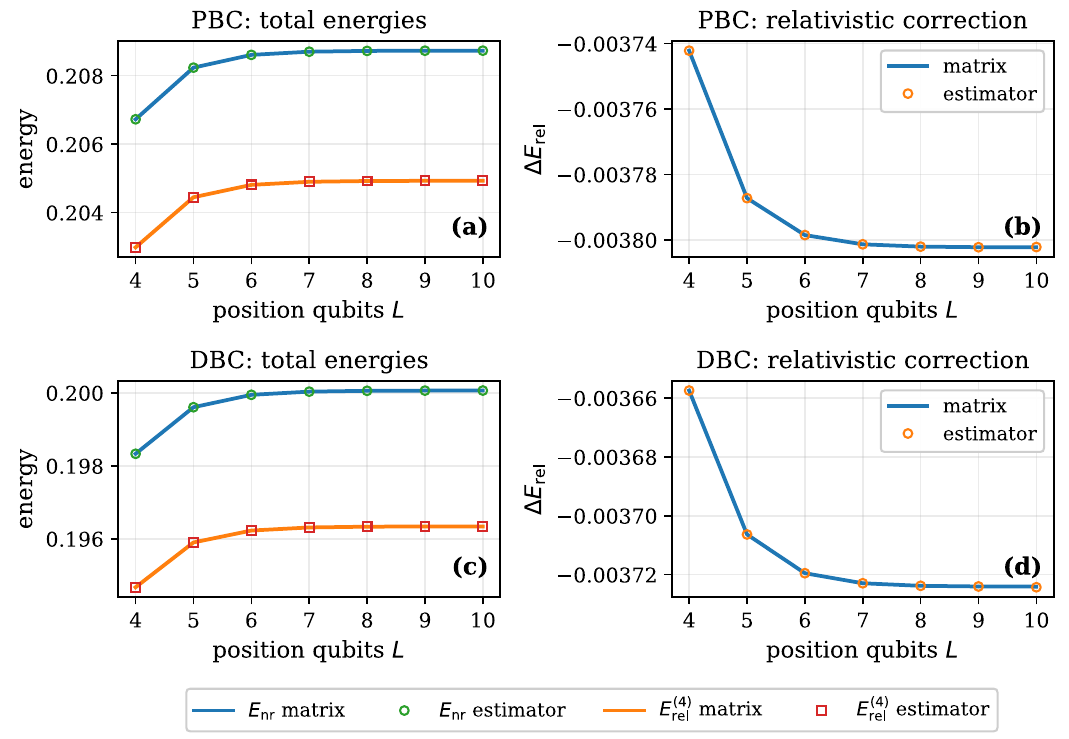}
\caption{Smooth-potential benchmark. The graphs (a) and (b) show the PBC results for the periodic potential in Eq.~(\ref{eq:pbc_smooth_potential_validation}). The graphs (c) and (d) show the DBC results for the confining potential in Eq.~(\ref{eq:dbc_smooth_potential_validation}). The first-order total energy $E_{\rm 1st}$, the perturbatively corrected energy $E_{\rm rel}^{(4)}$, and the relativistic correction $\Delta E_{\rm rel}$ are plotted as functions of the number of position qubits $L$. The solid lines denote direct matrix evaluation, while the open markers denote the translation-plus-boundary estimator reconstruction. }
\label{fig:smooth_potential_benchmark}
\end{figure*}

Fig.~\ref{fig:smooth_potential_benchmark} shows the results for $R=10$, $c=2$, $V_0=0.5$, and $\omega=0.4$. The grid sizes range from $L=4$ to $L=10$. The lines show direct matrix evaluation and the markers show the estimator reconstruction. The two are visually indistinguishable on the scale of the plot; the largest absolute difference in $E_{\rm rel}^{(4)}$ over all displayed PBC and DBC data points is $3.9 \times 10^{-7}$. The maximum value of the diagnostic ratio $\langle\hat{P}^{2}_{\tau}\rangle/(m^2 c^2)$ is approximately $5.1 \times 10^{-2}$, so the benchmark remains in the weakly relativistic regime for which Eq.~(\ref{eq:validation_T4}) is intended.

\subsection{Estimator-level finite-shot validation}\label{subsec:finite_shot_validation}

Finally, we test the statistical scaling of the measurement primitives. The translation-moment measurement produces a binary outcome $z^{(l)} \in \{-1,+1\}$ with mean $m_l={\rm Re}\langle\hat{A}^{l}\rangle$, whereas each boundary-overlap component is estimated from Bernoulli projection measurements. Therefore, for a fixed state, the measurement contribution to the energy error is expected to scale as
\begin{eqnarray}
\epsilon_{\rm meas}(M) = O\bigl(M^{-1/2}\bigr),
\label{eq:shot_noise_scaling_expected}
\end{eqnarray}
where $M$ is the number of shots allocated to each measurement primitive.

We quantify this scaling using the root-mean-square error
\begin{eqnarray}
{\rm RMSE}(M) = \sqrt{ \mathbb{E}\left[ \left( \widetilde{T}_{\tau}(M) - T_{\tau} \right)^2 \right]}
\label{eq:rmse_definition_validation}
\end{eqnarray}
for the kinetic estimator. The expectation is estimated by Monte Carlo sampling of the corresponding binary or Bernoulli measurement outcomes. The states are the smooth-potential ground states at $L=5$, using the same physical parameters as in Sec.~\ref{subsec:smooth_potential_validation}. For each value of $M$, we use $700$ independent Monte Carlo repetitions.

\begin{figure}[t]
\centering
\includegraphics[width=0.7\linewidth]{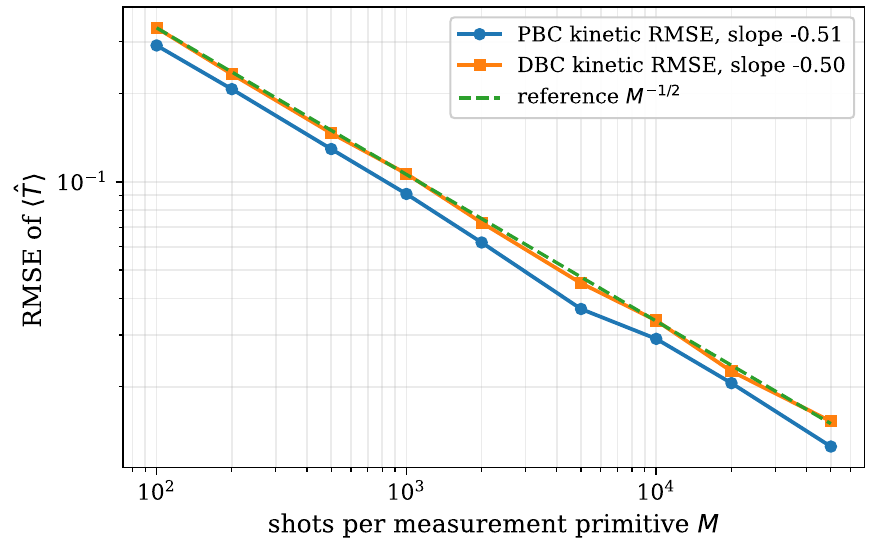}
\caption{Finite-shot scaling of the kinetic-energy estimators. The plotted RMSE is computed from repeated Monte Carlo sampling of the measurement primitives. The PBC estimator samples the translation moments $\langle\hat{A}\rangle$ and $\langle\hat{A}^{2}\rangle$ and the DBC estimator additionally is used for the boundary probabilities entering $\hat{E}_0$, $\hat{E}_1$, $\hat{E}_2$, and $\hat{E}_0^2$. The dashed line indicates the reference $M^{-1/2}$ scaling.
}
\label{fig:finite_shot_validation}
\end{figure}

Fig.~\ref{fig:finite_shot_validation} shows the resulting RMSE curves. The fitted slopes are $-0.51$ for the PBC kinetic estimator and $-0.50$ for the DBC kinetic estimator, in agreement with the expected $M^{-1/2}$ behaviour. The DBC curve has a slightly larger prefactor because it includes additional boundary-overlap probabilities, but it follows the same asymptotic statistical scaling as the translation-moment estimator.

Table~\ref{tab:validation_observable_summary} summarises the observable set used in the validation. The table emphasises that the leading relativistic correction enlarges the measured translation moments from $l=1$ to $l=1,2$ while DBC adds only boundary-local probability estimators.

\begin{table}[ht]
\centering
\begin{tabular}{c c c}
Quantity 									& Required measurement 		& Used-for \\
\hline
${\rm Re}\langle\hat{A}\rangle$ 				& Controlled translation 		& $\hat{P}^{2}$, $\hat{P}^{4}$ \\
${\rm Re}\langle\hat{A}^{2}\rangle$ 				& Controlled double translation 	& $\hat{P}^{4}$ \\
$\langle\hat{E}_0\rangle$ 						& Boundary overlap 			& DBC $\hat{P}^{2}$, DBC $\hat{P}^{4}$ \\
$\langle\hat{E}_1\rangle$, $\langle\hat{E}_2\rangle$ & Boundary overlap 			& DBC $\hat{P}^{4}$ \\
$\langle\hat{E}_0^2\rangle$ 					& Endpoint probabilities 		& DBC $\hat{P}^{4}$
\end{tabular}
\caption{Observable set required to reconstruct the kinetic moments through $\hat{P}^{4}$. The DBC observables are additional boundary-local terms; they are not needed under PBC.}
\label{tab:validation_observable_summary}
\end{table}

\section{Discussion}\label{sec:discussion}

We have presented a first-quantized methodology for relativistic quantum simulation under periodic and Dirichlet boundary conditions, denoted here by PBC and DBC. The main object of the framework was the finite-domain relativistic Hamiltonian. Starting from the positive-energy kinetic operator, the weakly relativistic simulator required the boundary-consistent moments $\langle\hat{P}^{2}\rangle$ and $\langle\hat{P}^{4}\rangle$. PBC and DBC were handled within the same position-register encoding, but they corresponded to different lattice momentum operators and therefore to different estimator formulas.

For PBC, the finite-domain structure was naturally represented by the unitary cyclic translation on the qubit register. The first kinetic energy term was reconstructed from $\mathrm{Re}\langle\hat{A}\rangle$, and the leading relativistic correction additionally required $\mathrm{Re}\langle\hat{A}^{2}\rangle$. For DBC, the target Hamiltonian was the open-chain finite-difference Hamiltonian on the interior grid. We implemented this target by decomposing it into a cyclic translation contribution plus boundary-local terms that cancel the unphysical wrap-around coupling. The additional DBC information was therefore limited to endpoint probabilities and endpoint or near-endpoint coherences associated with the expectation values of $\hat{E}_0$, $\hat{E}_1$, $\hat{E}_2$, and $\hat{E}_0^2$, rather than an extensive set of measurements over the whole grid.

The validation results support the interpretation. The PBC free-particle benchmark reproduced the lattice dispersion relation and its convergence toward the continuum relativistic energy. The DBC infinite-square-well benchmark showed that the boundary-local reconstruction gives the same $\hat{P}_{\rm D}^{2}$ and $\hat{P}_{\rm D}^{4}$ moments as direct open-chain matrix evaluation and as the analytic sine spectrum. The smooth-potential benchmarks confirmed that the full energy workflow, including diagonal potential sampling, agrees with direct matrix evaluation for nontrivial ground states under both boundary conditions. Finally, the finite-shot tests showed the expected $M^{-1/2}$ scaling for the kinetic-energy estimator, with the DBC case carrying a larger but controlled prefactor because of the additional boundary-overlap measurements.

Our formulation is intended for weakly relativistic regimes in which the positive-energy square-root kinetic operator is accurately represented by the low-momentum expansion through $\hat{P}^{4}$. In practical applications, the prepared lattice state should therefore have momentum support small compared with $mc$, and cutoff-scale grid modes should not dominate the energy estimate. This is a physical limitation of the perturbative relativistic approximation, not of the estimator identities themselves. The identities remain exact for the lattice operators defined in Sec.~\ref{sec:lattice_hamiltonian}; the expansion determines when the reconstructed quantity is a quantitatively accurate relativistic energy. Higher-order relativistic corrections and higher-order finite-difference stencils are summarized in \ref{app:higher_order}.

We should also highlight the importance of treating the boundary conditions as an essential part of quantum simulation. This perspective is particularly relevant for finite-domain systems, where the imposed boundary can directly shape the physical observables being simulated. Our framework showed that such the boundary awareness can be incorporated within a first-quantized quantum-simulation setting in a systematic and practical way. Beyond the specific examples considered here, the same viewpoint may be useful for more complex potentials, higher-dimensional geometries, and refined relativistic approximations. We therefore expect this boundary-aware formulations to play an important role in developing reliable relativistic quantum simulations of continuum systems.

\section*{Acknowledgments}

This work was supported by the Ministry of Science, ICT and Future Planning (MSIP) by the National Research Foundation of Korea (RS-2024-00432214, RS-2025-03532992, RS-2023-00281456, and RS-2023-NR119931) and the Institute of Information and Communications Technology Planning and Evaluation grant funded by the Korean government (RS-2019-II190003, “Research and Development of Core Technologies for Programming, Running, Implementing and Validating of Fault-Tolerant Quantum Computing System”). This work is also supported by the Grant No.~K25L5M2C2 at the Korea Institute of Science and Technology Information (KISTI).

\section*{Data availability statement} 

The data that support the findings of this study are openly available at the following URL/DOI: 

% ...

\appendix

\section{Grid conventions and dimensional analysis}\label{app:grid_conventions}

This appendix fixes the relation between the physical coordinate, the dimensionless coordinate, and the lattice parameters used in the main text. We include these details because the finite-difference momentum operators contain explicit powers of the grid spacing, and a dimensionally consistent convention is essential when relativistic parameters are introduced.

Let $x_{\rm phys} \in [0,R]$ denote the physical coordinate and let $y = {x_{\rm phys}}/{R}$ be the corresponding dimensionless coordinate. Then, we have
\begin{eqnarray}
&& \frac{d}{d x_{\rm phys}} = \frac{1}{R}\frac{d}{dy}, \nonumber \\
&& \hat{p} = -i\hbar\frac{d}{d x_{\rm phys}} = -i\frac{\hbar}{R}\frac{d}{dy} .
\label{eq:app_derivative_rescaling}
\end{eqnarray}
Thus, a finite-difference expression written in the dimensionless variable must carry the physical factor $R^{-1}$ in the momentum and $R^{-2}$ in the momentum squared.

For PBC, we use $N=2^L$ grid points on the circle,
\begin{eqnarray}
y_j^{({\rm P})} &=& \frac{j}{N}, \nonumber \\
x_j^{({\rm P})} &=& R y_j^{({\rm P})} = \frac{j R}{N}, \nonumber \\
\Delta_{\rm P} &=& \frac{R}{N}.
\label{eq:app_pbc_grid_physical}
\end{eqnarray}
For DBC, we use $N$ interior grid points in the interval,
\begin{eqnarray}
y_j^{({\rm D})} &=& \frac{j+1}{N+1}, \nonumber \\
x_j^{({\rm D})} &=& R y_j^{({\rm D})} = \frac{(j+1)R}{N+1}, \nonumber \\
\Delta_{\rm D} &=& \frac{R}{N+1}.
\label{eq:app_dbc_grid_physical}
\end{eqnarray}
The virtual boundary values are located at $y=0$ and $y=1$, or equivalently, at $x_{\rm phys}=0$ and $x_{\rm phys}=R$, and satisfy
\begin{eqnarray}
\psi(0)=0,
\quad
\psi(R)=0.
\label{eq:app_dbc_boundary_values}
\end{eqnarray}
The computational states $\ket{0}$ and $\ket{N-1}$ in the DBC register
therefore correspond to the first and last interior points, not to the
physical boundary points.

The physical second-derivative finite difference is
\begin{eqnarray}
\hbar^2\frac{d^2}{d x_{\rm phys}^2} ~\longrightarrow~ \frac{\hbar^2}{\Delta_\tau^2},
\quad
\tau \in \{{\rm P},{\rm D}\}.
\label{eq:app_physical_finite_difference}
\end{eqnarray}
Equivalently, if a dimensionless grid spacing $\delta_\tau = \frac{\Delta_\tau}{R}$ is used, then
\begin{eqnarray}
\frac{\hbar^2}{\Delta_\tau^2} = \frac{\hbar^2}{R^2\delta_\tau^2}.
\label{eq:app_delta_conversion}
\end{eqnarray}
This is the point at which dimensional consistency is most easily lost: a grid spacing in the dimensionless coordinate must not be used as if it were a physical length.

The dimensionless parameter controlling the size of the lattice momentum scale relative to $mc$ is
\begin{eqnarray}
\mu_\tau = \frac{\lambda_m}{\Delta_\tau} = \frac{\hbar}{mc\Delta_\tau},
\quad
\tau \in \{{\rm P},{\rm D}\},
\label{eq:app_mu_tau_definition}
\end{eqnarray}
where $\lambda_m$ is the reduced Compton wavelength: $\lambda_m =  \hbar /(mc)$.

Then, we have
\begin{eqnarray}
\frac{\hbar^2}{\Delta_\tau^2} &=& m^2 c^2 \mu_\tau^2.
\label{eq:app_mu_tau_powers}
\end{eqnarray}
If one instead works with a dimensionless coordinate $y$ and defines a notation, such as, $L_m=\lambda_m/\delta_\tau$, then $\lambda_m$ in that expression must itself be dimensionless, namely,
\begin{eqnarray}
\lambda_m^{({\rm dimless})} = \frac{\lambda_m}{R}.
\label{eq:app_dimensionless_lambda}
\end{eqnarray}
With this convention,
\begin{eqnarray}
L_{m,\tau} = \frac{\lambda_m^{({\rm dimless})}}{\delta_\tau} = \frac{\lambda_m/R}{\Delta_\tau/R} = \frac{\lambda_m}{\Delta_\tau} = \mu_\tau.
\label{eq:app_Lm_conversion}
\end{eqnarray}
Thus, the dimensionless lattice parameter that appears in the momentum operator is the ratio of the physical Compton wavelength to the physical grid spacing.

The positive-energy relativistic kinetic operator can be written in terms of
\begin{eqnarray}
\hat{Z}_\tau = \frac{\hat{P}_\tau^2}{m^2 c^2}
\label{eq:app_Z_tau_definition}
\end{eqnarray}
as
\begin{eqnarray}
\hat{T}_{\rm rel, \tau} = mc^2\left[ (\hat{\mathds{1}} + \hat{Z}_\tau)^{1/2} - \hat{\mathds{1}} \right].
\label{eq:app_exact_lattice_relativistic_operator}
\end{eqnarray}
The expansion used in the main text is
\begin{eqnarray}
\sqrt{1+z} - 1 = \frac{z}{2} - \frac{z^2}{8} + \frac{z^3}{16} + O(z^4),
\label{eq:app_scalar_rel_expansion}
\end{eqnarray}
and hence,
\begin{eqnarray}
\hat{T}_\tau = \frac{\hat{P}_\tau^2}{2m} - \frac{\hat{P}_\tau^4}{8 m^3 c^2}.
\label{eq:app_T4_tau_again}
\end{eqnarray}
For $z \geq 0$, Taylor's theorem gives
\begin{eqnarray}
0 \leq (1+z)^{1/2} - 1 - \frac{z}{2} + \frac{z^2}{8} \leq \frac{z^3}{16}.
\label{eq:app_remainder_scalar_bound}
\end{eqnarray}
Therefore, by functional calculus for the positive semidefinite operator $\hat{Z}_\tau$,
\begin{eqnarray}
0 \leq \hat{T}_{\rm rel,\tau}-\hat{T}_\tau \leq \frac{1}{16 m^5 c^4}\hat{P}_\tau^6,
\label{eq:app_operator_remainder_bound}
\end{eqnarray}
where $\hat{P}_\tau^6 \equiv (\hat{P}_\tau^2)^3$. Consequently, for any normalised lattice state $\ket{\psi}$,
\begin{eqnarray}
0 \leq \bra{\psi}\hat{T}_{\rm rel,\tau}\ket{\psi} - \bra{\psi}\hat{T}_\tau\ket{\psi} \leq \frac{\bra{\psi}\hat{P}_\tau^6\ket{\psi}}{16 m^5 c^4}.
\label{eq:app_state_remainder_bound}
\end{eqnarray}
This bound is not used as a numerical error estimate in the main text, but it makes explicit the low-momentum character of the perturbative approximation.

\section{Derivation of the PBC and DBC momentum moments}\label{app:boundary_moments}

Here we derive the momentum-moment identities used in Secs.~\ref{sec:lattice_hamiltonian} and~\ref{sec:quantum_estimators}. We keep the derivation algebraic, because the same manipulations extend directly to higher-order moments.

\subsection{PBC translation moments}\label{appsub:pbc_translation_moments}

Let $\hat{K} = \hat{A}+\hat{A}^{\dagger}-2\hat{\mathds{1}}$, where $\hat{A}$ is the cyclic translation satisfying $\hat{A}^{\dagger}=\hat{A}^{-1}$ and $\hat{A}^{N}=\hat{\mathds{1}}$. The PBC momentum-squared operator is
\begin{eqnarray}
\hat{P}_{\rm P}^{2} = -\frac{\hbar^2}{\Delta_{\rm P}^{2}}\hat{K}.
\label{eq:app_P2P_K}
\end{eqnarray}
Thus, for any positive integer $r$,
\begin{eqnarray}
\hat{P}_{\rm P}^{2r} \equiv (\hat{P}_{\rm P}^{2})^r = \left( -\frac{\hbar^2}{\Delta_{\rm P}^{2}} \right)^r \hat{K}^r.
\label{eq:app_general_PBC_moment_K}
\end{eqnarray}
Since $\hat{A}$ and $\hat{A}^{\dagger}$ commute, the binomial expansion gives
\begin{eqnarray}
\hat{K}^r &=& \sum_{q=0}^{r} \binom{r}{q} (-2)^{r-q} (\hat{A}+\hat{A}^{\dagger})^q \nonumber \\
	&=& \sum_{q=0}^{r} \binom{r}{q} (-2)^{r-q} \sum_{s=0}^{q} \binom{q}{s} \hat{A}^{q-2s}.
\label{eq:app_general_K_expansion}
\end{eqnarray}
Here, the negative powers are interpreted as powers of $\hat{A}^{\dagger}$, i.e., $\hat{A}^{-\ell}=\hat{A}^{\dagger\ell}$. Taking the expectation in a state $\ket{\psi}$ therefore reduces every PBC moment to translation moments $\langle \hat{A}^{\ell}\rangle$.

For $r=1$,
\begin{eqnarray}
\hat{P}_{\rm P}^{2} = -\frac{\hbar^2}{\Delta_{\rm P}^{2}} (\hat{A}+\hat{A}^{\dagger}-2\hat{\mathds{1}}),
\label{eq:app_PBC_P2_explicit}
\end{eqnarray}
which is Eq.~(\ref{eq:P2_pbc_operator}). For $r=2$,
\begin{eqnarray}
\hat{K}^2 &=& (\hat{A}+\hat{A}^{\dagger}-2\hat{\mathds{1}})^2 = \hat{A}^2 + \hat{A}^{\dagger 2} - 4\hat{A} - 4\hat{A}^{\dagger} + 6\hat{\mathds{1}},
\label{eq:app_K_squared}
\end{eqnarray}
and therefore
\begin{eqnarray}
\hat{P}_{\rm P}^{4} = \frac{\hbar^4}{\Delta_{\rm P}^{4}} \left( \hat{A}^2 + \hat{A}^{\dagger 2} - 4\hat{A} - 4\hat{A}^{\dagger} + 6\hat{\mathds{1}} \right),
\label{eq:app_PBC_P4_explicit}
\end{eqnarray}
which is Eq.~(\ref{eq:P4_pbc_operator}).

It is useful to introduce
\begin{eqnarray}
m_\ell = {\rm Re}\bra{\psi}\hat{A}^{\ell}\ket{\psi}.
\label{eq:app_m_ell_definition}
\end{eqnarray}
Then,
\begin{eqnarray}
\langle\hat{P}_{\rm P}^{2}\rangle &=& \frac{2\hbar^2}{\Delta_{\rm P}^{2}}(1 - m_1), \nonumber \\
\langle\hat{P}_{\rm P}^{4}\rangle &=& \frac{\hbar^4}{\Delta_{\rm P}^{4}}(2 m_2 - 8m_1 + 6),
\label{eq:app_PBC_m1_m2}
\end{eqnarray}
as used in Sec.~\ref{sec:quantum_estimators}.

\subsection{DBC as an open-chain correction to the cyclic shift}\label{appsub:dbc_open_chain_derivation}

Define the non-unitary open-chain forward shift
\begin{eqnarray}
\hat{B} = \sum_{j=0}^{N-2}\ketbra{j+1}{j}.
\label{eq:app_B_open_shift}
\end{eqnarray}
It shifts all interior grid points forward except that it does not connect $\ket{N-1}$ back to $\ket{0}$. The cyclic shift can be decomposed as
\begin{eqnarray}
\hat{A} &=& \hat{B}+\ketbra{0}{N-1}, \nonumber \\
\hat{A}^{\dagger} &=& \hat{B}^{\dagger}+\ketbra{N-1}{0}.
\label{eq:app_A_B_decomposition}
\end{eqnarray}
With $\hat{E}_0 = \ketbra{N-1}{0}+\ketbra{0}{N-1}$, we have
\begin{eqnarray}
\hat{A}+\hat{A}^{\dagger} - \hat{E}_0 = \hat{B} + \hat{B}^{\dagger}.
\label{eq:app_A_minus_E0_open_shift}
\end{eqnarray}
Therefore, we verify
\begin{eqnarray}
\hat{A} + \hat{A}^{\dagger} - 2\hat{\mathds{1}} - \hat{E}_0 = \hat{B} + \hat{B}^{\dagger} - 2\hat{\mathds{1}},
\label{eq:app_open_laplacian_identity}
\end{eqnarray}
which is exactly the open-chain second-difference operator. The DBC momentum operator is thus
\begin{eqnarray}
\hat{P}_{\rm D}^{2} &=& -\frac{\hbar^2}{\Delta_{\rm D}^{2}} (\hat{B}+\hat{B}^{\dagger}-2\hat{\mathds{1}}) = -\frac{\hbar^2}{\Delta_{\rm D}^{2}}(\hat{K}-\hat{E}_0),
\label{eq:app_P2D_open_chain_identity}
\end{eqnarray}
where $\hat{K}=\hat{A}+\hat{A}^{\dagger}-2\hat{\mathds{1}}$. Equivalently,
\begin{eqnarray}
\hat{P}_{\rm D}^{2} = \hat{P}_{\rm cyc}^{2}(\Delta_{\rm D}) + \frac{\hbar^2}{\Delta_{\rm D}^{2}}\hat{E}_0.
\label{eq:app_P2D_cyclic_plus_E0}
\end{eqnarray}
By taking the expectation values, we have Eq.~(\ref{eq:P2_dbc_expectation}).

\subsection{Full derivation of the DBC fourth moment}\label{appsub:dbc_fourth_moment_derivation}

The fourth DBC moment is the square of $\hat{P}_{\rm D}^{2}$:
\begin{eqnarray}
\hat{P}_{\rm D}^{4} = \left( \hat{P}_{\rm D}^{2}\right)^2 = \left(\frac{\hbar^2}{\Delta_{\rm D}^{2}}\right)^2(\hat{K} - \hat{E}_0)^2.
\label{eq:app_P4D_start}
\end{eqnarray}
Expanding the non-commuting product gives
\begin{eqnarray}
(\hat{K}-\hat{E}_0)^2 = \hat{K}^2 - \hat{K}\hat{E}_0 - \hat{E}_0\hat{K} + \hat{E}_0^2.
\label{eq:app_Q_squared_expansion}
\end{eqnarray}
The first term is the cyclic fourth-moment operator evaluated with the DBC grid spacing. The mixed terms are
\begin{eqnarray}
\hat{K}\hat{E}_0 + \hat{E}_0\hat{K} &=& (\hat{A}+\hat{A}^{\dagger}-2\hat{\mathds{1}})\hat{E}_0 + \hat{E}_0(\hat{A}+\hat{A}^{\dagger}-2\hat{\mathds{1}}) \nonumber\\
	&=& \hat{A}\hat{E}_0 + \hat{A}^{\dagger}\hat{E}_0 + \hat{E}_0\hat{A}  + \hat{E}_0\hat{A}^{\dagger} - 4\hat{E}_0.
\label{eq:app_mixed_terms_expanded}
\end{eqnarray}
Using $\hat{E}_1 = \hat{A}\hat{E}_0 + \hat{E}_0\hat{A}^{\dagger}$ and $\hat{E}_2 = \hat{E}_0\hat{A} + \hat{A}^{\dagger}\hat{E}_0$, we find the following:
\begin{eqnarray}
\hat{K}\hat{E}_0 + \hat{E}_0\hat{K} =  \hat{E}_1 + \hat{E}_2 - 4\hat{E}_0.
\label{eq:app_mixed_terms_E1_E2}
\end{eqnarray}
Therefore, we have
\begin{eqnarray}
\hat{P}_{\rm D}^{4} = \hat{P}_{\rm cyc}^{4}(\Delta_{\rm D}) + \frac{\hbar^4}{\Delta_{\rm D}^{4}} \left( 4\hat{E}_0 - \hat{E}_1 - \hat{E}_2 + \hat{E}_0^2 \right),
\label{eq:app_P4D_full_operator}
\end{eqnarray}
and taking expectation values gives Eq.~(\ref{eq:P4_dbc_expectation}).

For $N>2$, the explicit endpoint forms follow directly. First, we have
\begin{eqnarray}
\hat{E}_0^2 &=& \bigl( \ketbra{N-1}{0} + \ketbra{0}{N-1} \bigr)^2 = \ketbra{N-1}{N-1} + \ketbra{0}{0}.
\label{eq:app_E0_square_derivation}
\end{eqnarray}
Next,
\begin{eqnarray}
\hat{A}\hat{E}_0 &=& \ketbra{0}{0} + \ketbra{1}{N-1}, \nonumber \\
\hat{E}_0\hat{A}^{\dagger} &=& \ketbra{N-1}{1} + \ketbra{0}{0},
\label{eq:app_AE0_E0Adag_derivation}
\end{eqnarray}
which gives
\begin{eqnarray}
\hat{E}_1 = 2\ketbra{0}{0} + \ketbra{1}{N-1} + \ketbra{N-1}{1}.
\label{eq:app_E1_explicit_derivation}
\end{eqnarray}
Similarly,
\begin{eqnarray}
\hat{E}_0\hat{A} &=& \ketbra{N-1}{N-1} + \ketbra{0}{N-2}, \nonumber \\
\hat{A}^{\dagger}\hat{E}_0 &=& \ketbra{N-2}{0} + \ketbra{N-1}{N-1},
\label{eq:app_E0A_AdagE0_derivation}
\end{eqnarray}
and hence, we have
\begin{eqnarray}
\hat{E}_2 &=& 2\ketbra{N-1}{N-1} + \ketbra{0}{N-2} + \ketbra{N-2}{0}.
\label{eq:app_E2_explicit_derivation}
\end{eqnarray}
These expressions show explicitly why the DBC correction through $\hat{P}^{4}$ is boundary-local.

\section{Shot-noise and measurement-cost analysis}\label{app:shot_noise}

This appendix derives the leading finite-shot scaling of the kinetic-energy estimators. All estimates below assume the independent shot sets for different measurement primitives. The correlated or shared-shot estimators can reduce the constant prefactors but do not change the $M^{-1/2}$ scaling of the root-mean-square error.

\subsection{Translation-moment variance}\label{appsub:translation_variance}

The Hadamard estimator for $m_l = {\rm Re}\langle \hat{A}^l \rangle$ produces a binary random variable $Z_l \in \{-1,+1\}$ with
\begin{eqnarray}
\mathbb{E}[Z_l] = m_l,
\quad
{\rm Var}(Z_l) = 1 - m_l^2.
\label{eq:app_binary_translation_variance}
\end{eqnarray}
With $M_l$ shots,
\begin{eqnarray}
\widetilde{m}_l = \frac{1}{M_l}\sum_{r=1}^{M_l}Z_{l,r}, \nonumber \\
{\rm Var}(\widetilde{m}_l) = \frac{1 - m_l^2}{M_l} \leq \frac{1}{M_l}.
\label{eq:app_m_l_variance}
\end{eqnarray}

For PBC, define
\begin{eqnarray}
\alpha_{\rm P} = \frac{\hbar^2}{m\Delta_{\rm P}^{2}},
\quad
\beta_{\rm P} = \frac{\hbar^4}{m^3 c^2\Delta_{\rm P}^{4}}.
\label{eq:app_aP_bP}
\end{eqnarray}
Then, Eq.~(\ref{eq:T_pbc_moments}) can be written as
\begin{eqnarray}
T_{\rm P} = \alpha_{\rm P}(1-m_1) - \frac{ \beta_{\rm P}}{4} (m_2 - 4m_1 + 3).
\label{eq:app_TP_linear_moments}
\end{eqnarray}
The linear error propagation formula gives
\begin{eqnarray}
{\rm Var}\bigl(\widetilde{T}_{\rm P} \bigr) &=& \left(-\alpha_{\rm P}+\beta_{\rm P}\right)^2{\rm Var}(\widetilde{m}_1) + \frac{1}{16} \beta_{\rm P}^2{\rm Var}(\widetilde{m}_2),
\label{eq:app_var_TP_exact_linear}
\end{eqnarray}
assuming independent estimates of $m_1$ and $m_2$. Hence,
\begin{eqnarray}
{\rm Var}\bigl(\widetilde{T}_{\rm P} \bigr) \leq \frac{\left(-\alpha_{\rm P}+\beta_{\rm P}\right)^2}{M_1} + \frac{\beta_{\rm P}^2}{16 M_2}.
\label{eq:app_var_TP_bound}
\end{eqnarray}
If $M_1$ and $M_2$ are both proportional to a common shot budget $M$, then ${\rm RMSE}(\widetilde{T}_{\rm P})=O(M^{-1/2})$.

\subsection{Boundary-overlap variance}\label{appsub:boundary_variance}

A probability $P$ estimated from $M$ independent projection measurements has Bernoulli variance
\begin{eqnarray}
{\rm Var}(\widetilde{P}) = \frac{P(1-P)}{M} \leq \frac{1}{4M}.
\label{eq:app_bernoulli_probability_variance}
\end{eqnarray}
For $B_{fg} = 2P_{fg}^{+}-P_f-P_g$, using the independent estimates of $P_{fg}^{+}$, $P_f$, and $P_g$ with $M$ shots each gives
\begin{eqnarray}
{\rm Var}(\widetilde{B}_{fg}) &=& 4{\rm Var}(\widetilde{P}_{fg}^{+}) + {\rm Var}(\widetilde{P}_{f}) + {\rm Var}(\widetilde{P}_{g}) \leq \frac{3}{2M}.
\label{eq:app_Bfg_variance_bound}
\end{eqnarray}
The bound is conservative because the endpoint probabilities can often be obtained from the same computational-basis samples used for the potential estimator.

For DBC, define
\begin{eqnarray}
a_{\rm D} = \frac{\hbar^2}{m\Delta_{\rm D}^{2}},
\quad
d_{\rm D} = \frac{\hbar^2}{2m\Delta_{\rm D}^{2}},
\quad
e_{\rm D} = \frac{\hbar^4}{8 m^3 c^2\Delta_{\rm D}^{4}}.
\label{eq:app_aD_dD_bD}
\end{eqnarray}
Combining Eqs.~(\ref{eq:P2_dbc_estimator_moments}) and (\ref{eq:P4_dbc_estimator_moments}), the DBC kinetic estimator can be written as
\begin{eqnarray}
\hspace{-2.5cm} T_{\rm D}  = a_{\rm D}(1-m_1) + d_{\rm D}b_0 - e_{\rm D}(2m_2 - 8m_1 + 6 + 4b_0 - b_1 - b_2 + b_{00}) \nonumber\\
\hspace{-2cm} 	= a_{\rm D}-6 e_{\rm D} + (-a_{\rm D} + 8 e_{\rm D})m_1 - 2 e_{\rm D}m_2 + (d_{\rm D} - 4 e_{\rm D})b_0 + e_{\rm D}b_1 + e_{\rm D} b_2 - e_{\rm D}b_{00}.
\label{eq:app_TD_linear_form}
\end{eqnarray}
The first term in the second line is a constant shift with respect to the sampled quantities, and therefore it does not contribute to the variance. Assuming independent estimates of the non-constant terms,
\begin{eqnarray}
\hspace{-1.5cm} {\rm Var}\left(\widetilde{T}_{\rm D}^{(4)}\right) &=& (-a_{\rm D} + 8e_{\rm D})^2{\rm Var}(\widetilde{m}_1) + (2e_{\rm D})^2{\rm Var}(\widetilde{m}_2) \nonumber\\
\hspace{-1.5cm} &&+ (d_{\rm D}-4e_{\rm D})^2{\rm Var}(\widetilde{b}_0) + e_{\rm D}^2{\rm Var}(\widetilde{b}_1) + e_{\rm D}^2{\rm Var}(\widetilde{b}_2) + e_{\rm D}^2{\rm Var}(\widetilde{b}_{00}).
\label{eq:app_var_TD_linear}
\end{eqnarray}
Each variance on the right-hand side is $O(M^{-1})$ when $M$ shots are allocated to each primitive, so
\begin{eqnarray}
{\rm RMSE}(\widetilde{T}_{\rm D}^{(4)}) = O(M^{-1/2}).
\label{eq:app_TD_RMSE_scaling}
\end{eqnarray}
The finite-shot simulations in Sec.~\ref{subsec:finite_shot_validation} confirm this scaling.

\subsection{Potential-energy sampling variance}\label{appsub:potential_variance_appendix}

To keep the notation consistent with Proposition~\ref{prop:potential_sampling}, we reserve $\hat{V}_{\tau}$ for the potential operator and use $\widetilde{V}_{\tau}$ for the Monte Carlo estimator. Let $J_1,\ldots,J_M$ be independent position-basis samples with ${\rm Pr}(J_r=j)=\abs{c_j}^2$. Define the sample-level random variables and their average by
\begin{eqnarray}
X_{r,\tau} := V\bigl(x_{J_r}^{(\tau)}\bigr),
\quad
\widetilde{V}_{\tau} := \frac{1}{M}\sum_{r=1}^{M}X_{r,\tau},
\label{eq:app_potential_random_variable}
\end{eqnarray}
where $r=1,\ldots,M$. Then,
\begin{eqnarray}
\mathbb{E}\bigl[\widetilde{V}_{\tau}\bigr] &=& \frac{1}{M}\sum_{r=1}^{M}\mathbb{E}[X_{r,\tau}] = \sum_{j=0}^{N-1}V\bigl(x_j^{(\tau)}\bigr)\abs{c_j}^2 = \langle\hat{V}_{\tau}\rangle, \nonumber \\
{\rm Var}(\widetilde{V}_\tau) &=& \frac{1}{M^2}\sum_{r=1}^{M}{\rm Var}(X_{r,\tau}) = \frac{{\rm Var}(X_{\tau})}{M}.
\label{eq:app_potential_variance_exact}
\end{eqnarray}
Here, the last equality uses the iid assumption, and $X_{\tau}$ denotes a generic copy of the single-shot variable $X_{r,\tau}$.
If $V(x_j^{(\tau)}) \in [V_{\min}, V_{\max}]$, then
\begin{eqnarray}
{\rm Var}(\widetilde{V}_\tau) \leq \frac{(V_{\max} - V_{\min})^2}{4M}.
\label{eq:app_potential_variance_bound_again}
\end{eqnarray}
The total-energy variance is obtained by adding the kinetic and potential variance contributions when independent shot sets are used.

\subsection{Measurement settings and circuit scaling}\label{appsub:measurement_cost}

Through the leading relativistic correction, the PBC kinetic estimator requires two translation-moment settings: $l=1$ and $l=2$. The DBC kinetic estimator uses the same two translation settings and adds the boundary-overlap settings for the three pairs
\begin{eqnarray}
(0,N-1), \quad (1,N-1), \quad (0,N-2),
\label{eq:app_boundary_pairs}
\end{eqnarray}
plus the endpoint probabilities, which can be obtained from computational-basis samples. Thus, up to $\hat{P}^{4}$, the number of boundary-specific DBC settings is independent of $N$.

The controlled implementation of $\hat{A}^l$ is a controlled modular increment by $l$ on an $L$-qubit register~\cite{Vedral1996}. For fixed $l$, a ripple-carry construction has gate count scaling polynomially and typically linearly in $L$ up to the architecture-dependent constants and ancilla choices. The QFT-based adders give an alternative implementation with a different constant and connectivity profile. The estimator framework does not depend on a particular adder design. It only requires that controlled cyclic translations can be implemented or otherwise estimated as unitary register operations.

\section{Boundary-overlap measurement identities}\label{app:boundary_overlap}

This appendix gives the probability identities underlying the boundary-overlap estimators. Let
\begin{eqnarray}
\ket{\psi} = \sum_{j=0}^{N-1}c_j\ket{j},
\quad
c_j=\braket{j}{\psi}.
\label{eq:app_overlap_state}
\end{eqnarray}
For two distinct computational basis states $\ket{f}$ and $\ket{g}$, define
\begin{eqnarray}
P_f &=& \abs{c_f}^2, \nonumber \\
P_g &=& \abs{c_g}^2, \nonumber \\
P_{fg}^{+} &=& \abs{\langle{s_{fg}^{+}}|{\psi}\rangle}^2,
\label{eq:app_Pf_Pg_Pplus}
\end{eqnarray}
where $|s_{fg}^{+}\rangle = ({\ket{f}+\ket{g}})/{\sqrt{2}}$. Then, we have
\begin{eqnarray}
P_{fg}^{+} &=& \frac{1}{2}\abs{c_f + c_g}^2 = \frac{1}{2} \left( \abs{c_f}^2 + \abs{c_g}^2 + c_f^{\ast} c_g + c_g^{\ast} c_f \right).
\label{eq:app_Pplus_expanded}
\end{eqnarray}
Consequently,
\begin{eqnarray}
2P_{fg}^{+}-P_f-P_g &=& c_f^{\ast} c_g + c_g^{\ast} c_f = \bra{\psi} \bigl( \ketbra{f}{g} + \ketbra{g}{f} \bigr) \ket{\psi}.
\label{eq:app_real_overlap_identity}
\end{eqnarray}
The factor of two multiplying $P_{fg}^{+}$ is a direct consequence of the normalisation of $|{s_{fg}^{+}}\rangle$.

The same construction can access imaginary coherences when needed. Define the phase-dependent superposition
\begin{eqnarray}
P_{fg}^{\phi} = \abs{\langle{s_{fg}^{\phi}}|{\psi}\rangle}^2,
\label{eq:app_phase_superposition}
\end{eqnarray}
where $|{s_{fg}^{\phi}}\rangle = {(\ket{f}+e^{i\phi}\ket{g})}/{\sqrt{2}}$. Then, we have
\begin{eqnarray}
&& 2P_{fg}^{\phi} - P_f - P_g  = e^{-i\phi}c_f^{\ast}c_g + e^{i\phi} c_g^{\ast} c_f = 2{\rm Re}\bigl(e^{-i\phi} c_f^{\ast} c_g\bigr).
\label{eq:app_phase_overlap_identity}
\end{eqnarray}
The choice $\phi=0$ gives the real coherence in Eq.~(\ref{eq:app_real_overlap_identity}), while $\phi=\pi/2$ gives
\begin{eqnarray}
&& 2P_{fg}^{\pi/2} - P_f - P_g  = 2{\rm Im}(c_f^{\ast}c_g) = \bra{\psi}\bigl( -i\ketbra{f}{g} + i\ketbra{g}{f} \bigr)\ket{\psi}.
\label{eq:app_imag_overlap_identity}
\end{eqnarray}
The DBC kinetic estimators in the main text require only the real coherence because the boundary operators $\hat{E}_0$, $\hat{E}_1$, and $\hat{E}_2$ are Hermitian combinations of off-diagonal projectors.

Using the shorthand
\begin{eqnarray}
B_{f, g} = 2P_{fg}^{+} - P_f - P_g,
\label{eq:app_Bfg_again}
\end{eqnarray}
the DBC boundary terms through $\hat{P}^{4}$ are
\begin{eqnarray}
\langle\hat{E}_0\rangle &=& B_{0,N-1}, \nonumber \\
\langle\hat{E}_0^2\rangle &=& P_0 + P_{N-1}, \nonumber \\
\langle\hat{E}_1\rangle &=& 2P_0 + B_{1,N-1}, \nonumber \\
\langle\hat{E}_2\rangle &=& 2P_{N-1} + B_{0,N-2}.
\label{eq:app_Es_Bfg}
\end{eqnarray}
Thus, the boundary correction can be obtained from endpoint probabilities and three real two-state coherences.

\section{Numerical methods}\label{app:numerical_methods}

This appendix describes the numerical procedures used for the validation figures in Sec.~\ref{sec:validation}. All numerical simulations use dense linear algebra because the purpose is operator validation rather than large-scale classical performance.

\subsection{Matrix construction}\label{appsub:matrix_construction}

For a given $L$, $N=2^L$. The cyclic translation matrix is
\begin{eqnarray}
(\hat{A})_{j,k} = \delta_{j,k+1 \ {\rm mod} \ N},
\label{eq:app_A_matrix_elements}
\end{eqnarray}
where row and column indices run from $0$ to $N-1$. The PBC momentum matrices are
\begin{eqnarray}
\hat{P}_{\rm P}^{2} = -\frac{\hbar^2}{\Delta_{\rm P}^{2}} (\hat{A}+\hat{A}^{\dagger}-2\hat{\mathds{1}}),
\quad
\hat{P}_{\rm P}^{4} = (\hat{P}_{\rm P}^{2})^2.
\label{eq:app_PBC_matrix_construction}
\end{eqnarray}
For DBC, the open-chain shift is
\begin{eqnarray} 
(\hat{B})_{j,k} = \delta_{j,k+1}
\quad
(0 \leq k \leq N-2),
\label{eq:app_B_matrix_elements}
\end{eqnarray}
with no wrap-around element. The DBC matrices are
\begin{eqnarray}
\hat{P}_{\rm D}^{2} = -\frac{\hbar^2}{\Delta_{\rm D}^{2}}(\hat{B}+\hat{B}^{\dagger}-2\hat{\mathds{1}}),
\quad
\hat{P}_{\rm D}^{4} = (\hat{P}_{\rm D}^{2})^2.
\label{eq:app_DBC_matrix_construction}
\end{eqnarray}
Equivalently, the DBC matrices can be constructed from the cyclic shift using $\hat{E}_0$ as in Eq.~(\ref{eq:app_P2D_cyclic_plus_E0}). Both constructions are numerically identical up to floating-point precision.

The lattice potential is the diagonal matrix
\begin{eqnarray}
\hat{V}_{\tau} = \sum_{j=0}^{N-1}V\bigl(x_j^{(\tau)}\bigr)\ketbra{j}{j},
\quad
\tau\in\{{\rm P},{\rm D}\}.
\label{eq:app_potential_matrix_construction}
\end{eqnarray}
The first-order total Hamiltonian used to generate the smooth-potential ground
states is
\begin{eqnarray}
\hat{H}^{tot}_{{\rm 1st},\tau} = \frac{\hat{P}_{\tau}^{2}}{2m}+\hat{V}_{\tau}.
\label{eq:app_Hnr_matrix}
\end{eqnarray}
The ground state is obtained by exact diagonalisation of the Hermitian matrix $\hat{H}_{\tau,{\rm nr}}$.

\subsection{Estimator reconstruction from a state vector}\label{appsub:estimator_reconstruction_state_vector}

Given a normalised state vector $\vec{c}=(c_0,\ldots,c_{N-1})^{T}$, the exact translation moments are computed as $m_l={\rm Re}(\vec{c}^{\,T}\hat{A}^l \vec{c}\,)$. To avoid ambiguity in implementation, this can be written component-wise as
\begin{eqnarray}
m_l = {\rm Re}\left( \sum_{j=0}^{N-1} c_j^{\ast}c_{j-l \ {\rm mod} \ N} \right) \quad (l=1,2),
\label{eq:app_numerical_ml}
\end{eqnarray}
for the convention $\hat{A}\ket{j}=\ket{j+1 \ {\rm mod} \ N}$. The PBC estimator reconstruction is then
\begin{eqnarray}
\langle\hat{P}_{\rm P}^{2}\rangle_{\rm est} &=& \frac{2\hbar^2}{\Delta_{\rm P}^{2}}(1-m_1), \nonumber \\
\langle\hat{P}_{\rm P}^{4}\rangle_{\rm est} &=& \frac{2\hbar^4}{\Delta_{\rm P}^{4}}(m_2-4m_1+3).
\label{eq:app_PBC_estimator_reconstruction}
\end{eqnarray}
For DBC, the boundary quantities are reconstructed as
\begin{eqnarray}
b_0 &=& 2{\rm Re}(c_0^{\ast}c_{N-1}), \nonumber \\
b_{00} &=& \abs{c_0}^2 + \abs{c_{N-1}}^2, \nonumber \\
b_1 &=& 2\abs{c_0}^2 + 2{\rm Re}(c_1^{\ast}c_{N-1}), \nonumber \\
b_2 &=& 2\abs{c_{N-1}}^2 + 2{\rm Re}(c_0^{\ast} c_{N-2}).
\label{eq:app_bs_numerical}
\end{eqnarray}
The DBC estimator reconstruction is
\begin{eqnarray}
\langle\hat{P}_{\rm D}^{2}\rangle_{\rm est} &=& \frac{2\hbar^2}{\Delta_{\rm D}^{2}}(1-m_1) + \frac{\hbar^2}{\Delta_{\rm D}^{2}}b_0, \nonumber \\
\langle\hat{P}_{\rm D}^{4}\rangle_{\rm est} &=& \frac{\hbar^4}{\Delta_{\rm D}^{4}} \left(2m_2-8m_1+6+4b_0-b_1-b_2+b_{00}\right).
\label{eq:app_P24D_estimator_reconstruction}
\end{eqnarray}
These expressions are compared with direct matrix expectations
\begin{eqnarray}
\langle\hat{P}_{\tau}^{2}\rangle_{\rm mat} =  \vec{c}^{\,T} \hat{P}_{\tau}^{2} \vec{c},
\quad
\langle\hat{P}_{\tau}^{4}\rangle_{\rm mat} = \vec{c}^{\,T} \hat{P}_{\tau}^{4} \vec{c}.
\label{eq:app_matrix_expectations}
\end{eqnarray}

\subsection{Benchmark parameters}\label{appsub:benchmark_parameters}

The PBC free-particle benchmark uses the Fourier mode $n=1$ with $R=10$, $\hbar=m=1$, $c=2$, and $L=3,4,\ldots,10$. For each $L$, the lattice momentum eigenvalue is
\begin{eqnarray}
p_{\rm lat}^{2}(n) = \frac{4\hbar^2}{\Delta_{\rm P}^{2}} \sin^2\left(\frac{\pi n}{N}\right),
\label{eq:app_pbc_free_plat}
\end{eqnarray}
and the continuum value is
\begin{eqnarray}
p_{\rm cont}^{2}(n) = \left(\frac{2\pi n\hbar}{R}\right)^2.
\label{eq:app_pbc_free_pcont}
\end{eqnarray}
The energies plotted in Fig.~\ref{fig:pbc_free_particle_dispersion} are $T_{\rm cont}$, $T_{\rm lat}$, and $T_{\rm pert}$ as defined in Eq.~(\ref{eq:T_validation}).

The DBC square-well benchmark uses $N=64$, $R=10$, $\hbar=m=1$, and compares modes $s=1,\ldots,20$. The analytic sine eigenvectors are
\begin{eqnarray}
\psi_s(j) = \sqrt{\frac{2}{N+1}} \sin\left(\frac{\pi s(j+1)}{N+1}\right),
\label{eq:app_dbc_sine_numerical}
\end{eqnarray}
with eigenvalues
\begin{eqnarray}
p_{{\rm D},{\rm lat}}^{2}(s) = \frac{4\hbar^2}{\Delta_{\rm D}^{2}} \sin^2\left(\frac{\pi s}{2(N+1)}\right).
\label{eq:app_dbc_square_eigenvalue}
\end{eqnarray}

The smooth-potential benchmark uses $R=10$, $\hbar=m=1$, $c=2$, and $L=4,5,\ldots,10$. For PBC, the potential is
\begin{eqnarray}
V_{\rm P}(x) = V_0\left(1-\cos\left(\frac{2\pi x}{R}\right)\right),
\label{eq:app_pbc_smooth_potential}
\end{eqnarray}
for $V_0=0.5$. For DBC, the potential is
\begin{eqnarray}
V_{\rm D}(x) = \frac{1}{2} m \omega^2\left(x-\frac{R}{2}\right)^2,
\label{eq:app_dbc_smooth_potential}
\end{eqnarray}
with $\omega=0.4$. For each grid, we diagonalise $\hat{H}^{tot}_{\tau,{\rm nr}}$ and evaluate $E_{{\rm nr},\tau}$, $\Delta E_{{\rm rel},\tau}$, and $E_{{\rm rel},\tau}^{(4)}$ on the ground state, as in Eq.~(\ref{eq:E_validations}).

The finite-shot benchmark uses the smooth-potential ground states at $L=5$ with the same physical parameters, allocates $M$ shots to each measurement primitive, and estimates the RMSE over independent Monte Carlo repetitions.
The discretized ground states $\ket{\psi_{0,\tau}}$ are represented by trigonometric and Gaussian wavefunctions for PBC and DBC, respectively.

\subsection{Error metrics}\label{appsub:error_metrics}

For a momentum moment, we use the relative discrepancy
\begin{eqnarray}
\epsilon_{P^{2r}} = \frac{ \abs{\langle\hat{P}_{\tau}^{2r}\rangle_{\rm est} - \langle\hat{P}_{\tau}^{2r}\rangle_{\rm mat}}}{\max\left\{ \abs{\langle\hat{P}_{\tau}^{2r}\rangle_{\rm mat}}, \; \epsilon_{\rm floor}\right\}},
\quad
r=1,2,
\label{eq:app_relative_moment_error}
\end{eqnarray}
where $\epsilon_{\rm floor}$ is a small numerical floor used only to avoid a zero denominator. For total energies, we use the absolute reconstruction error
\begin{eqnarray}
\epsilon_E = \abs{E_{\rm est} - E_{\rm mat}}.
\label{eq:app_absolute_energy_error}
\end{eqnarray}
For the weakly relativistic diagnostic, we monitor
\begin{eqnarray}
\eta_{\tau} = \frac{\langle\hat{P}_{\tau}^{2}\rangle}{m^2c^2}.
\label{eq:app_eta_tau_diagnostic}
\end{eqnarray}
This state-dependent quantity does not bound all higher moments, but it is a useful indicator that the low-momentum expansion is being applied in the intended regime. 
A more conservative diagnostic is the spectral-support condition
\begin{eqnarray}
\eta_{\tau, {\rm max}} = \max_{\lambda \in {\rm supp}(\psi)}\frac{\lambda}{m^2 c^2},
\label{eq:app_spectral_eta}
\end{eqnarray}
where $\lambda$ runs over eigenvalues of $\hat{P}_{\tau}^{2}$ that have non-negligible weight in the state $\ket{\psi}$.

\section{Higher-order extensions and additional benchmarks}\label{app:higher_order}

This appendix records two natural extensions of the framework: higher-order relativistic moments and higher-order finite-difference stencils. These extensions are not required for the main validation, but they clarify how the present construction generalises.

\subsection{Including the $\hat{P}^{6}$ relativistic correction}\label{appsub:P6_extension}

The next term in the positive-energy expansion is
\begin{eqnarray}
\hat{T}^{(6)}_\tau = \frac{\hat{P}_\tau^2}{2m} - \frac{\hat{P}_\tau^4}{8 m^3 c^2} + \frac{\hat{P}_\tau^6}{16 m^5 c^4},
\label{eq:app_T6_definition}
\end{eqnarray}
where $\hat{P}_\tau^6=(\hat{P}_\tau^2)^3$. For PBC, using $\hat{K}=\hat{A}+\hat{A}^{\dagger}-2\hat{\mathds{1}}$ and $\hat{P}_{\rm P}^{2}=-(\hbar^2/\Delta_{\rm P}^{2})\hat{K}$, one finds
\begin{eqnarray}
\hat{P}_{\rm P}^{6} =
\frac{\hbar^6}{\Delta_{\rm P}^{6}} \left[ 20\hat{\mathds{1}} - 15\left(\hat{A}+\hat{A}^{\dagger}\right) + 6\left(\hat{A}^{2}+\hat{A}^{\dagger 2}\right) - \left( \hat{A}^{3}+\hat{A}^{\dagger 3} \right) \right].
\label{eq:app_PBC_P6_operator}
\end{eqnarray}
Hence,
\begin{eqnarray}
\langle\hat{P}_{\rm P}^{6}\rangle = \frac{\hbar^6}{\Delta_{\rm P}^{6}} \left( 20 - 30m_1 + 12m_2 - 2m_3 \right),
\label{eq:app_PBC_P6_moments}
\end{eqnarray}
where $m_3={\rm Re}\langle\hat{A}^{3}\rangle$. Thus, the $\hat{P}^{6}$ term requires one additional translation moment under PBC.

For DBC,
\begin{eqnarray}
\hat{P}_{\rm D}^{6} = -\frac{\hbar^6}{\Delta_{\rm D}^{6}}\left(\hat{K}-\hat{E}_0\right)^3.
\label{eq:app_DBC_P6_formal}
\end{eqnarray}
Because $\hat{K}$ and $\hat{E}_0$ do not commute, the expansion contains all ordered products of $\hat{K}$ and $\hat{E}_0$. The resulting correction terms remain local near the endpoints, but they involve a wider boundary stencil than the $\hat{P}^{4}$ correction. In particular, terms can connect endpoint basis states to sites up to three lattice steps away. This is the natural higher-order analogue of the $\hat{E}_0$, $\hat{E}_1$, and $\hat{E}_2$ terms appearing at fourth order.

\subsection{Higher-order finite-difference stencils}\label{appsub:higher_order_stencils}

The main text uses the standard second-order finite-difference approximation to $-\partial_x^2$. A fourth-order accurate PBC stencil for the momentum squared is
\begin{eqnarray}
\hat{P}_{\rm P,4th}^{2} = \frac{\hbar^2}{\Delta_{\rm P}^{2}} \left[ \frac{5}{2}\hat{\mathds{1}} - \frac{4}{3}\left(\hat{A}+\hat{A}^{\dagger}\right) + \frac{1}{12}\left(\hat{A}^{2}+\hat{A}^{\dagger 2}\right) \right].
\label{eq:app_fourth_order_stencil}
\end{eqnarray}
This reduces the finite-grid discretization error for smooth states but requires additional translation moments even for the first-order relativistic kinetic energy term. Squaring this operator to obtain $\hat{P}^{4}$ introduces moments up to $\hat{A}^{4}$. Under DBC, the same stencil requires the boundary corrections that remove wrap-around couplings associated with both $\hat{A}$ and $\hat{A}^{2}$. Thus, higher-order stencils trade discretization accuracy for a larger but still structured observable set.





\begin{thebibliography}{}
\bibitem{NoteLatticeMoments} On a finite lattice, these are replaced by moments of a lattice momentum operator.
\end{thebibliography}
\end{document}